\documentclass[runningheads]{llncs}
\usepackage{graphicx}
\usepackage{xcolor}
\usepackage{amsmath, amssymb}
\usepackage[utf8]{inputenc}
\usepackage{algorithmicx}
\usepackage{algpseudocode}
\usepackage{enumerate}   
\usepackage{multirow}
\usepackage{algorithm}
\usepackage{verbatim}
\usepackage{mathtools}
\usepackage{float}
\usepackage{cite}
\begin{document}
\title{The application of $\sigma$-LFSR in Key-Dependent Feedback Configuration for Word-Oriented Stream Ciphers}
%
%
\author{Subrata Nandi*\inst{1}\orcidID{0000-0001-7140-0608}\and Srinivasan Krishnaswamy\inst{2}\orcidID{0000-0001-9714-7399}\and Behrouz Zolfaghari\inst{3}\orcidID{0000-0001-6691-0988
} \and Pinaki Mitra\inst{4}\orcidID{0000-0002-8254-8234}}

\authorrunning{S. Nandi et al.}
%
\institute{* Corresponding author,
	CSE Department, Indian Institute of Technology Guwahati, Gwa 781039, India \email{subrata.nandi@iitg.ac.in} \and
	EEE Department, Indian Institute of Technology Guwahati, Gwa 781039, India
	\email{srinikris@iitg.ac.in}\\
	\and
	CSE Department, Indian Institute of Technology Guwahati, Gwa 781039, India
	\email{zolfaghari@iitg.ac.in}\\
	\and
	CSE Department, Indian Institute of Technology Guwahati, Gwa 781039, India
	\email{pinaki@iitg.ac.in}
}

\maketitle              
\begin{abstract}
	
 In this paper, we propose and evaluate a  method for generating key-dependent feedback configurations (KDFC) for $\sigma$-LFSRs.  $\sigma$-LFSRs with such configurations  can be applied to any stream cipher that uses a word-based LFSR. Here, a configuration generation algorithm uses the secret key(K) and the initialization vector (IV) to generate a feedback configuration. We have mathematically analysed the feedback configurations generated by this method. As a test case, we have applied this method on SNOW 2.0  and have studied its impact on resistance to various attacks. Further, we have also tested the generated keystream for randomness and have briefly described its implementation and the challenges involved in the same.

\end{abstract}
\section{Introduction}

Stream ciphers are used in a variety of applications \cite{Appl-01, Appl-02}. LFSRs (Linear Feedback Shift Register) are widely used as building blocks in stream ciphers because of their simple construction and easy implementation.

Word based LFSRs were introduced to efficiently use the structure of modern word based processors (\cite{rose1999t,berbain2008sosemanuk, xiu2011zuc, ekdahl2019new}). Such LFSRs are used in a variety of stream ciphers, most notably in the SNOW series of stream ciphers. A $\sigma$-LFSR is a word based LFSR configuration that was introduced in \cite{zeng2007high}. An important property of this configuration is that there are multiple feedback functions corresponding to a given characteristic polynomial of the state transition matrix(\cite{krishnaswamy2011number}). The number of such configurations was conjectured in (\cite{zeng2007high}). This conjecture was constructively proved in {\cite{krishnaswamy2011number}}

The knowledge of the feedback function plays a critical role in most attacks on LSFR based stream ciphers. These include algebraic attacks, correlation attacks and distinguishing attacks (\cite{billet2005resistance,watanabe2003distinguishing,nyberg2006improved,lee2008cryptanalysis,zhang2015fast}). Therefore, hiding the feedback function of the LFSR could potentially increase the security of such schemes. In this paper we try doing this by making the feedback function dependent on the secret key. The resulting configuration is called the $\sigma$-KDFC (Key Dependent Feedback Configuration). The proposed method for obtaining the feedback function from the secret key utilizes the algorithm given in \cite{krishnaswamy2011number,krishnaswamy2014number}.The feedback gains thus obtained are highly non-linear functions of the secret key. Further, the number of iterations in this algorithm can be adjusted depending on the available computing power. As an example, we study the interconnection of the $\sigma$-KDFC with the finite state machine (FSM) of SNOW-2. We use empirical tests to verify the randomness of the keystream generated by this scheme. Further, we analyse the scheme for security against various kinds of attacks.

 In this paper,  $\mathbb{F}_{2^n}$ denotes a finite field of cardinality $2^n$.  $\mathbb{F}_2^n$ denotes an $n$-dimensional vector space over $\mathbb{F}_2$. The $i^{th}$ row and $j^{th}$ column of a matrix $M \in \mathbb{F}^{n\times n}$ are denoted by $M[i,:]$ and $M[:,j]$ respectively. $M[i,j]$ denotes the $(i,j)$-th entry of the matrix $M$. The minor of $M[i,j]$ is denoted by $\mu(M[i,j])$. $\oplus $ and $+$ are used interchangeably to represent addition over $\mathbb{F}_2$. 

The rest of this paper is organized as follows. Section \ref{BasCon} introduces LFSRs, $\sigma$-LFSRs and some related concepts. Section \ref{Prop} examines $\sigma$-KDFC and its time complexity. Section \ref{AlgRes} presents the mathematical analysis on the algebraic degree of the parameters of the feedback function generated by $\sigma$-KDFC. Section \ref{CaseSt} discusses the the interconnection of $\sigma$-KDFC with the FSM of SNOW and its security against various cryptographic attacks . Section \ref{Conc} concludes the paper.

\section{LFSRs and $\sigma$-LFSRs}\label{BasCon}
An LFSR is an electronic circuit that implements a linear recurring relation (LRR)  of the form $x_{n+b} = a_{b-1}x_{n+b-1} + \cdots + a_0x_n$. It consists of a shift register with $b$ flip-flops and a linear feedback function which is typically implemented using XOR gates. The integer $b$ is called the length of the LFSR. The characteristic polynomial of an LFSR is a monic polynomial with the same coefficients as the LRR implemented by it. For example, the characteristic polynomial of an LFSR which implements the LRR given above is $x^b + a_{b-1}x^{b-1} + \cdots + a_0$. The period of the sequence generated by an LFSR of length $b$ is at most $2^{b}-1$.
Further, an LFSR generates a maximum-period sequence if its characteristic polynomial is primitive \cite{New-Pr-01}. The state vector of an LFSR at any time instant is a vector whose entries are the outputs of the delay blocks at tha time instant i.e $\mathbf{x_n} = [x_n,x_{n+1},\ldots,x_{n+b-1}]$. Two consecutive state vectors are related by the equation $\mathbf{x_{n+1}} = \mathbf{x_n}P_f$ where $P_f$ is the companion matrix of the characteristic polynomial.

In order to efficiently work with word based processors various word based LFSR designs have been proposed \cite{SNOW1-04,Dyn-Str-01,SNOW1-01,berbain2008sosemanuk,ekdahl2002new,ekdahl2019new,rose1999t}. These designs use multi input multi output delay blocks. One such design is the $\sigma$-LFSR shown in Figure \ref{SLFSR}. Here, the feedback gains are matrices and the implemented linear recurring relation is of the form
\begin{equation}
    \mathbf{x_{n+b}}  = \mathbf{x_{n+b-1}}B_{n+b-1} + \mathbf{x_{n+b-2}}B_{n+b-2} +\cdots + \mathbf{x_n}B_0
\end{equation}
where each $\mathbf{x_i} \in \mathbb{F}_2^m$ and $B_i \in \mathbb{F}_2^{m \times m}$. Here, each delay block has $m$-inputs and $m$-outputs and the $\sigma$-LFSR generates a sequence of vectors in $\mathbb{F}_2^m$ 
\begin{figure}[H]
	\centering
	\includegraphics[width=0.95\linewidth]{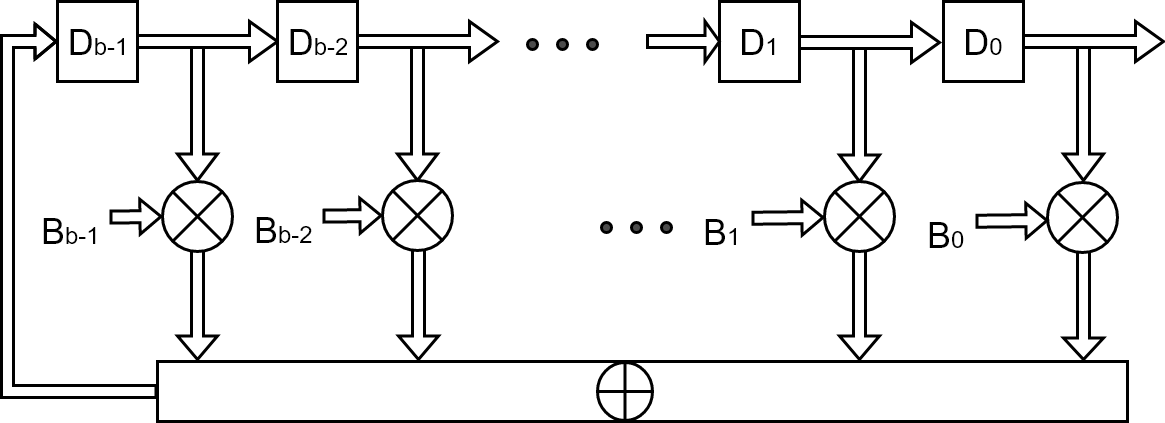}
	\caption{Block diagram of $\sigma$-LFSR}
	\label{SLFSR}
\end{figure}

The matrices $B_{0},B_{1},\cdots,B_{b-1}$  are referred to as the gain matrices of the $\sigma$-LFSR and the following matrix is defined as its configuration matrix.

\begin{equation}\label{SU}
C=
\begin{bmatrix}
0 & I & 0 & \cdots & 0 \\
0 & 0 & I & \cdots & 0 \\
\vdots & \vdots & \vdots & \cdots &\vdots\\
0 & 0 & 0 & \cdots & I\\
B_0 & B_1 & B_2 & \cdots & B_{b-1}
\end{bmatrix} \in \mathbb{F}_2^{mb \times mb}
\end{equation}
 where $0, I \in\mathbb{F}_2^{m\times m}$ are the all-zero and identity matrices respectively. We shall refer to the structure of this matrix as the $M$-companion structure. The characteristic polynomial of this  configuration matrix is known as the characteristic polynomial of the $\sigma$-LFSR. 
 
 If $\mathbf{x_n}$ is the output of the $\sigma$-LFSR at the $n$-th time instant, then its state vector at that time instant is defined as $\mathbf{\hat{x}_n} = [\mathbf{x_n},\mathbf{x_{n+1}},\ldots,\mathbf{x_{n+b-1}}]$. This vector is got by stacking the outputs of all the delay blocks at the $n$-th time instant. Two consecutive state vectors are related by the following equation:

\begin{equation}
    \mathbf{\hat{x}_{n+1}}=\mathbf{\hat{x}_n}C^T
\end{equation}

In the case of $\sigma$-LFSRs,  there are many possible feedback configurations having the same characteristic polynomial. For a given primitive polynomial, the number of such  configurations was conjectured in \cite{zeng2007high} to be the following,
\begin{equation}\label{conj}
N_P=\frac{|GL(m,\mathbb{F}_{2})|}{2^{m}-1}\times \frac{\phi(2^{mb}-1)}{mb}\times 2^{m(m-1)(b-1)}
\end{equation}
where \ref{conj}, $GL(m,\mathbb{F}_{2})$ is the general linear group of non-singular matrices $\in \mathbb{F}_2^{m \times m}$, and $\phi$ represents Euler's totient.
This conjecture has been proved in \cite{krishnaswamy2014number}. Moreover, this inductive proof is constructive and gives an algorithm  for calculating such feedback functions. 

In the following section, we shall use the algorithm given in \cite{krishnaswamy2014number} to develop a key dependent feedback configuration for the $\sigma$-LFSR

\section{$\sigma$-KDFC}\label{Prop}
Stream ciphers, like the SNOW series of ciphers, use word based LFSRs along with an FSM module. The feedback configuration of the LFSR in such schemes is publicly known. This feedback relation is used in most attacks on such schemes\cite{billet2005resistance,nyberg2006improved,ahmadi2009heuristic,zhang2015fast}. Therefore, the security of such schemes could potentially increase if the feedback function is made key dependent.

Before proceeding to our construction of a key dependent feedback configuration, we briefly describe the algorithm given in \cite{krishnaswamy2014number} which generates feedback configurations for $\sigma$-LFSRs with a given characteristic polynomial. Given a primitive polynomial $p_{mb}(x)$ having degree $mb$, the algorithm for calculating a feedback configuration for  a $\sigma$-LFSR with $b$ $m$-input $m$-output delay blocks is as follows:

\begin{algorithm}[H] \label{HPillai}
\caption{Configuration Matrix Generation}
    \begin{enumerate}
        \item Initialize $Y$ with a full rank matrix in $\mathbb{F}_2^{m \times m}$.
        \item Choose $mb-m-1$ primitive polynomials $p_{m}(x),p_{m+1}(x),\ldots, p_{mb-1}(x)$ having degrees $m.m+1,\ldots,mb-1$ respectively. (Note that the polynomial $p_{mb}(x)$ is given). Let $A_m,A_{m+1},\ldots,A_{mb}$ be the companion matrices of $p_{m}(x),p_{m+1}(x),\ldots, p_{mb}(x)$ respectively.
        \item For $i= 1$ to $mb-m$ update $Y$ as follows
         \begin{enumerate}
             \item Let $\ell$ be the unique integer less or equal to $m$ which is equivalent to $i$ mod $m$. Find a polynomial $f(x)$ such that $Y[\ell \textrm{ mod }m,:]A_{m+i-1} = (0,0,\ldots,1) \in \mathbb{F}_2^{m+i-1}$ and update $Y$ as follows
             \begin{eqnarray*}
             Y = Y*A_{m+i-1}
             \end{eqnarray*}
             
             \item For all $t \neq \ell \textrm{ mod }m $, $Y[t,:] = (Y[t,:],d_t)$ where $d_t$ is an element of  $\mathbb{F}_2$.(In this step all the rows  of $Y$, except the  one which is equivalent to $i$ mod $m$, are appended with random boolean numbers and their lengths are increased by 1.)
             \item If $t = \ell$, $Y[t,:]= (0,0,\ldots,1) \in \mathbb{F}_2^{m+i}$. 
         \end{enumerate}
         (At the end of each iteration, an extra column is added to $Y$ till $Y \in \mathbb{F}_2^{m \times mb}$.)
         \item Construct the following matrix $Q$\\.
         $Q \gets \begin{bmatrix}
		Y[0:,]\\
		Y[1:,]\\
		\vdots\\
		Y[m-1:,]\\
		Y[0:,]\times P_{mb}\\
		Y[1:,] \times P_{mb}\\
		\vdots\\
		Y[m-1:,] \times P_{mb}\\
		\vdots\\
		Y[0:,]\times \left(P_{mb}\right)^{b-1}\\
		Y[1:,] \times \left(P_{mb}\right)^{b-1}\\
		\vdots\\
		Y[m-1:,] \times \left(P_{mb}\right)^{b-1}
		\end{bmatrix}\in \mathbb{F}_2^{mb \times mb}$
		\item $C = Q \times P_{mb} \times Q^{-1}$ 
    \end{enumerate}
\end{algorithm}

 The matrix $C$ generated in the above algorithm is the configuration matrix of a $\sigma$-LFSR with characteristic polynomial $p_{mb}(x)$. As can be seen from Equation \ref{SU},  the last $m$ rows of this matrix contain the feedback gain matrices.  Each set of choices for the $d_t$s and the initial full rank matrix  result in a different feedback configuration.
 
In Step 3a, the coefficients of the polynomial $f(x)$ can be calculated by solving the linear equation $y\times K_t = (0,0,\ldots,1) \in \mathbb{F}_2^{m+i-1}$ for $y$, where $K_i$ is given by

\begin{eqnarray*}
K_i = \left[\begin{matrix} Y[\ell,:]\\Y[\ell,:]A_{m+i-1}\\\vdots\\Y[\ell,:]A_{m+i-1}^{m+i-2}\end{matrix}\right]
\end{eqnarray*}
In other words $f(x) = y(1) + y(2)x+\cdots+y(m+i-1)x^{m+i-2}$, where $y(j)$ is the $j$-th entry of the vector $y$.

Note that in every iteration of Step 3 in Algorithm 1, $m-1$ random numbers are appended to the rows of the matrix $Y$. In the proposed scheme, some of these numbers are derived from the secret key. As a consequence, the derived feedback configuration  is dependent on the secret key. We now proceed to look at this configuration in detail.

\begin{figure}[H]
	\centering
	\includegraphics[width=0.9\linewidth]{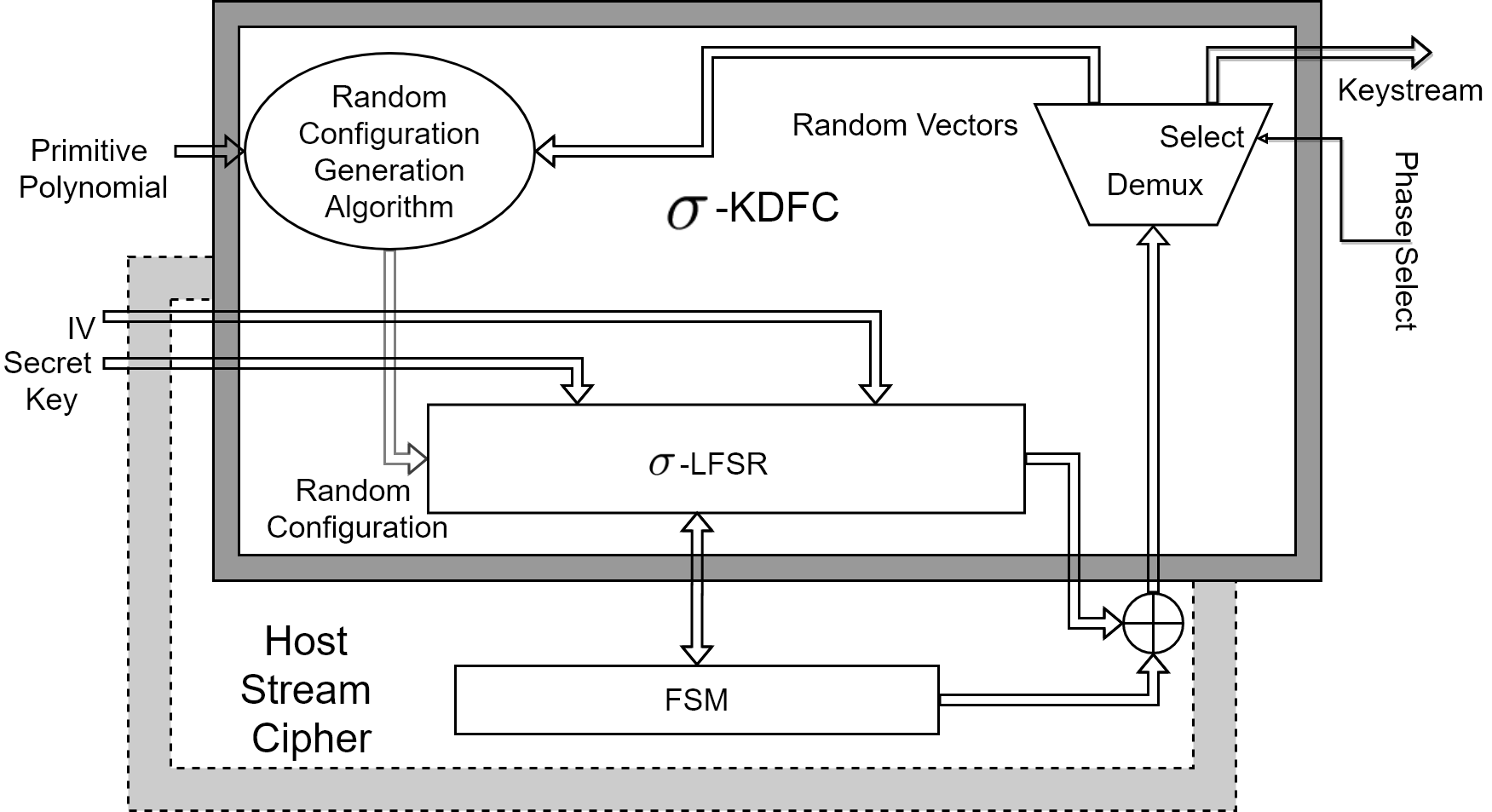}
	\caption{The Schematic of $\sigma$-KDFC}
	\label{Sigma-KDFC} 
\end{figure}
In order to create a keystream generator from the proposed $\sigma$-LFSR configuration, it can be connected to a Finite state machine which introduces non-linearity.
Figure \ref{Sigma-KDFC} shows the schematic of the proposed scheme along with its interconnection with an FSM. The scheme has an initialization phase wherein the feedback configuration of the $\sigma$-LFSR is calculated by running Algorithm 1. In order to reduce the time taken for initialization, Algorithm 1  is run offline (at some server) till $k$ iterations of step 3 and the resulting matrix $Y$ is made public. The number $k$ can be chosen depending on the computational capacity of the machine that hosts the $\sigma$-LFSR. The feedback configuration is calculated by running the remaining part of the algorithm in the initialization phase. In this phase there is no keystream generated at the output. The following subsection explains the initialization phase in detail.

\subsection {The Initialization Phase}
During the initialization phase the $\sigma$-LFSR has a publicly known feedback configuration. Further, the pre-calculated matrix $Y \in \mathbb{F}_2^{m \times (m+k)}$, and the primitive polynomials $p_{m+k+1}(x),p_{m+k+2}(x),\ldots,p_{mb}(x)$ are also publicly known.  The initial state of the $\sigma$-LFSR is derived from the secret key and the IV. (as is normally done in word based stream ciphers like SNOW). The $\sigma$-LFSR is run along with the FSM for $mb-m-k$ clock cycles. This generates $mb-m-k$ vectors in $\mathbb{F}_2^m$. This corresponds to the $mb-m-k$ remaining iterations of Step 3 in Algorithm 1.  

The remaining part of Algorithm 1 is now run. In each iteration of Step 3, the boolean numbers appended to the rows of the matrix $Y$ in Step 3(b) are the entries of the the corresponding vector. More precisely, in the $i$-th iteration of Step 3 that is run on the keystream generator, for $t \neq i+k$ mod $m$, the $t$-th row of $Y$ is appended with the $t$-th entry of the $i$-th vector that was generated.

The feedback gains of the $\sigma$-LFSR are now set according to the configuration matrix that is generated by Algorithm 1.

Once the feedback gains are set the $\sigma$-LFSR is run along with the FSM. The first $b$ vectors are discarded and the keystream starts from the $b+1$-th vector. The reason for doing this is that the initial state of the $\sigma$-LFSR with the new configuration is generated by the publicly known feedback configuration which is used in the initialization process.

\subsubsection{Time Complexity of the Initialization Phase}

Step 3(a) in Algorithm 1 involves solving a system of linear equations in less than $mb$ variables. This can be done with a time complexity of $\mathcal{O}((mb)^3)$ using Gaussian elimination. The time complexity of Step 3(b) is linear in $m$ while that of Step 3(c) is constant. Further,  Step 3 has $mb-m-k$ iterations. Therefore, if $k$ is chosen such that $mb-m-k$ is $\mathcal{O}(1)$, then the overall time complexity of Step 3 is $\mathcal{O}((mb)^3)$.  In Step 5, the matrix $C$ can be calculated by solving the linear system of equations $CQ = QP_{mb}$ for $C$. This can be done in $\mathcal{O}((mb)^4)$ using Gaussian elimination. Step 4 has a time complexity of $\mathcal{O}((mb)^3)$. Thus, the time complexity of the initialization phase is $\mathcal{O}((mb)^4)$.
\section{Algebraic Analysis of $\sigma$-KDFC}\label{AlgRes}

The entries of the feedback matrices, $B_0,B_1,\ldots,B_{b-1}$, calculated by the procedure given in the previous section are functions of the matrix $Y$ generated in Step 3 of Algorithm 1. The entries of $Y$ are in turn non-linear functions of the initial state of the $\sigma$-LFSR.

Note that the last row of $Y$ is always $e_{1}^n$. Let the first $m-1$ rows of $Y$ be $v_1,v_2,\ldots,v_{m-1}$. Let $\mathcal{U}$ be the set of variables that denote the entries in these rows. Therefore,

\begin{equation}\label{Lin}
      B_{k}(i,j)=f_{k(i,j)}(\mathcal{U})  \textrm{ for } 0\leq k \leq b-1 \textrm{ and } 1\leq i,j \leq m
\end{equation}
 where $f_{k(i,j)}$s are polynomial functions.
 
The algebraic degree of the configuration matrix, denoted by $\Theta$, is defined as follows
\begin{equation}\label{Deg}
 \Theta\left(C_{\mathcal{S}}\right)= \max_{k,i,j} \left( |f_{k(i,j)}(\mathcal{U})|\right)
\end{equation}
$\Theta$ can be considered as a measure of the algebraic resistance of $\sigma$-KDFC.  We now proceed to find a lower bound for $\Theta$.


The matrix $Q$ generated in Step 4 of Algorithm 1 is given as follows 
\begin{eqnarray}
Q = \left[\begin{matrix}v_1\\v_2\\\vdots\\ v_{m-1}\\e_1^n\\v_1P_{mb}\\\vdots\\v_{m-1}P_{mb}\\e_1^nP_{mb}\\\vdots\\ v_1P_{mb}^{b-1}\\\vdots\\v_{m-1}P_{mb}^{b-1}\\e_1P_{mb}^{b-1}\\ \end{matrix}\right]
\end{eqnarray}
where $P_{mb}$ is the companion matrix of the publicly known primitive characteristic polynomial of the $\sigma$-LFSR.
The configuration matrix $C$ is generated by the formula $C = Q \times P_{mb} \times Q^{-1}$. Since $Q$ is an invertable boolean matrix, the determinant of $Q$ is always 1. Therefore, $Q^{-1}=Q^{(a)}$ where $Q^{(a)}$ is the adjugate of $Q$. Moreover, since the elements of $Q$ belong to $\mathbb{F}_2$, the co-factors are equal to the minors of $Q$. 
 The rows of $Q$ can be permuted to get the following matrix $Q_P$

\begin{equation}\label{adj}
    Q^P=\left[\begin{matrix}e_1^n\\e_1^nP_{mb}\\\vdots\\e_1P_{mb}^{b-1}\\v_1\\v_1P_{mb}\\\vdots\\ v_1P_{mb}^{b-1}\\\vdots\\ v_{m-1}\\v_{m-1}P_{mb}\\\vdots\\v_{m-1}P_{mb}^{b-1}\\ \end{matrix}\right]=\left(\begin{array}{cccccccc}
0 & 0 &\cdots & 0 & 0 & \cdots & 0 & 1 \\
0 & 0 & \cdots & 0 & 0 & \cdots & 1 & * \\
\vdots & \vdots & \cdots & \vdots&\vdots\\
0 & 0 & \cdots & 0 & 1 & * &  * \\
v_{1,1} & v_{1,2} & \cdots & v_{1,n-b}& v_{1,n-b+1} & \cdots &v_{1,n-1} &v_{1,n}\\
v_{1,2} & v_{1,3} & \cdots & v_{1,n-b+1}& v_{1,n-b+2} & \cdots &v_{1,n} & * \\
\vdots & \vdots & \cdots & \vdots&\vdots\\
v_{1,b} & v_{1,b+1} & \cdots & v_{1,n} & v_{1,n+1} & \cdots & *&* \\
v_{2,1} & v_{2,2} &\cdots& v_{2,n-b} & v_{2,n-b+1} &\cdots & v_{2,n-1} & v_{2,n}\\
v_{2,2} & v_{2,3} &\cdots&v_{2,n-b+1} & v_{2,n-b+2} &\cdots &v_{2,n}& *\\
\vdots & \vdots & \cdots & \vdots&\vdots\\
v_{2,b}&v_{2,b+1}&\cdots& v_{2,n} & v_{2,n+1}&\cdots &*&*\\
\vdots & \vdots & \cdots & \vdots&\vdots\\
v_{m-1,1} &v_{m-1,2} &\cdots & v_{m-1,n-b}& v_{m-1,n-b+1} &\cdots &v_{m-1,n-1}&v_{m-1,n}\\
v_{m-1,2} &v_{m-1,3} &\cdots & v_{m-1,n-b+1}& v_{m-1,n-b+2}& \cdots & v_{m-1,n} & *\\
\vdots & \vdots & \cdots & \vdots & \vdots \\
v_{m-1,b} & v_{m-1,b+1} & \cdots & v_{m-1,n} & v_{m-1,n+1} & \cdots & * & *
\end{array}\right)
\end{equation}	 
where $*$s are linear combinations of the entries of the previous row. Note that $Q^{-1}$ can be got by permuting the rows of $Q_P^{-1}$

 The matrix $Q_P$  can be decomposed as follows into four sub-matrices $Q_1,Q_2,Q_3$ and $Q_4$:

\begin{equation}\label{Rew}
Q_P =\\
\begin{array}{c|c}
Q_1=\begin{pmatrix}
	0 & 0 &\cdots & 0   \\
	0 & 0 & \cdots & 0   \\
	\vdots & \vdots & \cdots& 	\vdots  \\
	0 & 0 & \cdots & 0  \\
\end{pmatrix}_{b \times n-b}
&
 Q_2=\begin{pmatrix}
0 & \cdots & 0 & 1\\
0 & \cdots & 1 & *\\
\vdots&\vdots\\
 1 & \cdots &* &  *
\end{pmatrix}_{b \times b}\\
\hline
  Q_3=\begin{pmatrix}
	v_{1,1} & v_{1,2} & \cdots & v_{1,mb-b}\\
	v_{1,2} & v_{1,3} & \cdots & v_{1,mb-b+1}\\
	\vdots & \vdots & \cdots \\
	v_{1,b} & v_{1,b+1} & \cdots & v_{1,mb}\\
	v_{2,1} & v_{2,2} &\cdots& v_{2,mb-b}\\
	v_{2,2} & v_{2,3} &\cdots&v_{2,mb-b+1}\\
	\vdots & \vdots & \cdots &	\vdots\\
	v_{2,b}&v_{2,b+1}&\cdots& v_{2,mb}\\
	\vdots & \vdots & \cdots & 	\vdots\\
	v_{m-1,1} &v_{m-1,2} &\cdots & v_{m-1,mb-b}\\
	v_{m-1,2} &v_{m-1,3} &\cdots & v_{m-1,mb-b+1}\\
	\vdots & \vdots & \cdots & 	\vdots\\
	v_{m-1,b} &v_{m-1,b+1} &\cdots & v_{m-1,mb}
\end{pmatrix}_{mb-b \times mb-b}
&
 Q_4=\begin{pmatrix}
 v_{1,mb-b+1} & \cdots &v_{1,mb-1} &v_{1,mb}\\
v_{1,mb-b+2} & \cdots &v_{1,mb} & * \\
\vdots&\vdots & \cdots \vdots\\
v_{1,mb+1} & \cdots & *&* \\
v_{2,mb-b+1} &\cdots & v_{2,mb-1} & v_{2,mb}\\
v_{2,mb-b+2} &\cdots &v_{2,mb}& *\\
\vdots&\vdots & \cdots & \vdots \\
v_{2,mb+1}&\cdots &*&*\\
 \vdots&\vdots& \cdots & \vdots\\
v_{m-1,mb-b+1} &\cdots &v_{m-1,mb-1}&v_{m-1,mb}\\
v_{m-1,mb-b+2}& \cdots & v_{m-1,mb} & *\\
\vdots & \vdots & \cdots & \vdots\\
v_{1,mb+1} & \cdots &*&*
\end{pmatrix}_{mb-b \times b}
\end{array}
\end{equation}

Since $Q_P$ is invertible, $det(Q_P)= det(Q_3) =1$.

Let  $\Gamma_k$ be the set of of polynomial functions of $\mathcal{U}$ variables with  degree  $k$. We now proceed to analyse some of the minors of $ Q_P$. 
\begin{lemma}\label{PUP}
For $1 \leq j \leq (mb-b)$, $\mu(Q_P[b,j]))\in    \Gamma_{mb-b} $
\end{lemma}
\begin{proof}
    For two matrices $A$ and $B$ with the same number of rows, let $[AB]_{p,q}$ be the matrix which is got by removing the $p^{th}$ column from $A$ and appending the $q^{th}$ column of $B$ to $A$. For $i=b$ and $1 \leq j \leq (mb-b)$, $\mu(Q_P[i,j]))$ is given by:
   \begin{equation}\label{Det}
    \mu(Q_P[i,j])=det([Q_3Q_4]_{j,1})
    \end{equation}
    Recall that, for a binary matrix $M \in \mathbb{F}_2^{mb \times mb}$, its determinant is given by the following formula,
    \begin{equation}
    det(M) = \sum_{f \in S_{mb}} \prod_{1 \leq i \leq n}M(i,f(i))
    \end{equation}
    where $S_{mb}$ is the set of permutations on $(1,2,\ldots,mb)$.
    Observe that the diagonal elements of $([Q_3Q_4])_{j,1}$ are distinct $v_{i,k}$s. Their product corresponds to the identity permutation in the determinant expansion formula for $[Q_3Q_4])_{j,1}$. The resultant monomial has degree $mb-b$.
    Further, this monomial will not occur as a result of any other permutation. Hence $det([Q_3Q_4]_{j,1})$ is always a polynomial of degree $mb-b$.

   \end{proof}

 \begin{lemma}\label{PUP1}
 If $1\leq i \leq b$ then
\begin{equation}\label{MinQ2}
     \mu(Q_P[i,j])=
    \begin{cases}
        det(Q_3)&i+j=mb+1\\
        0& i+j >= mb+1
    \end{cases}
\end{equation}
\end{lemma}

    \begin{proof}
    
    Observe that, for $1\leq i \leq b$ and $i+j = mb+1$, the $Q_P[i,j]$s are the anti-diagonal elements of $Q_2$. Clearly, the minors of these elements are all equal to the determinant of $Q_3$. As we have already seen, the invertibility of $Q_P$ implies that this determinant is always $1$. Therefore, $\mu(Q_P[i,j]) = 1$ when $i+j = mb+1$
   
    Note that, for $1\leq i \leq b$ and  $i+j > mb+1 $, the $Q_P[i,j]$s are the elements of $Q_2$ that are below the anti-diagonal. Observe that, if the row and column corresponding to such an element are removed from $Q_P$, then the first $b-1$ rows of the resulting matrix are always rank deficient. Therefore, the determinant of this matrix is always 0. Therefore, $\mu(Q_P[i,j]) = 0$.
 \end{proof}
    
\begin{lemma}\label{PUP2}
 If $b+1 \leq i \leq mb$ and $1 \leq j \leq n-b$, then $\mu(Q_P[i,j]) \in \Gamma_{mb-b-1}$.
\end{lemma} 
\begin{proof}    
   Observe that the elements of $Q_P$ considered in this lemma are elements of the sub-matrix $Q_3$. Therefore, $\mu(Q_P[i,j])$, for the range of $i$ and $j$ considered, is nothing but the determinant of the sub-matrix of $Q_3$  got by deleting the $i^{th}$ row and $j^{th}$ column of $Q_3$.
 The diagonal elements of such a sub-matrix are distinct $v_{i,j}s$. Their product will result in a monomial of degree $mb-b-1$. This corresponds to the identity permutation in the determinant expansion formula given by Equation \ref{Det}. Observe that no other permutation generates this monomial. Hence, the minor will always have a monomial of degree $mb-b-1$ . Therefore,  $\mu(Q_P[i,j] \in \Gamma_{mb-b-1}$.
\end{proof}

\begin{lemma}\label{PUP3}
If $b+1\leq i \leq n$ and $mb-b+1 \leq j \leq mb$, then $\mu(Q_P[i,j])=0$.
\end{lemma} 
 \begin{proof}
 The elements of $Q_P $ considered in this lemma are elements of the submatrix $Q_4$. Whenever the row and column corresponding to such an element is removed from $Q_P $, the rows of the submatrix $Q_2$ become linearly dependent. Therefore, the first $b$ rows of the resultant matrix are always rank deficient. Consequently, $\mu(Q_P[i,j])=0$. 
    
\end{proof}

For a given matrix $A$ with polynomial entries, let $\Theta(A)$ be the maximum degree among all the  entries of $A$. As there are $mb-b$ rows in $Q_P$ with variable entries, $\Theta(Q_P^{-1})\leq mb-b$. Therefore, we get the following as a consequence of Lemma \ref{PUP}.

\begin{equation}
   \Theta(Q^{-1}) =  \Theta(Q_P^{-1}) = mb-b
\end{equation}

Recall that the configuration matrix $\mathcal{C}$ is given by $QP_{mb}Q^{-1}$. We now use the above developed machinery to calculate $\Theta(\mathcal{C})$.

\begin{theorem}\label{IC}
 $\Theta(\mathcal C) \ge mb-b$ . 
\end{theorem}
\begin{proof}
Observe that the gain matrices $B_0,B_1\cdots,B_{b-1}$ appear in the last $m$ rows of $C_{\mathcal{S}}$. These rows are generated by multiplying the last $m$ rows $QP_{mb}$ with $Q^{-1}$.  The last $m$ rows of $QP_{mb}$ are as follows
\begin{equation}\label{LastRows}
\left(\begin{array}{rrrrrrrr}
0 & 0 & \cdots & 1 & * & \cdots & * & *\\  
v_{1,b+1} &v_{1,b+2} &\cdots & v_{1,mb}& * &\cdots & * & *\\
v_{2,b+1} &v_{2,b+2} &\cdots & v_{2,mb}& *& \cdots & * & *\\
\vdots & \vdots & \cdots & \vdots&\vdots& \cdots & \vdots & \vdots\\
v_{m-1,b+1} &v_{m-1,b+2} &\cdots & v_{m-1,mb}& * & \cdots &*&*
\end{array}\right)\in \mathbb{F}_2^{b \times mb}
\end{equation}
The element $C[mb-m+1,mb-m+1]$ is got by multiplying the $(mb-m+1)$-th row of $QP_{mb}$ with the $(mb-m+1)$-th column of $Q^{-1}$. Note that the $(mb-m+1)$-th column of $Q^{-1}$ is equal to the $b$-th column of $Q_P^{-1}$. As a consequence of Lemmas \ref{PUP} and \ref{PUP1}, this column has the following form.
\begin{equation}
Q^{-1}[:,mb-m+1]=(P_1,P_2,\cdots,P_{mb-b},1,0,\cdots,0)^T
\end{equation}
where $P_1,P_2,\cdots,P_{mb-b} \in \Gamma(mb-b)$. Therefore,
\begin{align*}
C[mb-b+1,mb-b+1]&=(\underbrace{0, 0 , \cdots,  1}_{(mb-b) \textrm{ entries}}, * , \cdots , * , *) \times (P_1,P_2,\ldots,P_{mb-b},1,0,\cdots,0)^T\\
&=P_{mb-b}
\end{align*}
 Hence, it is proved that $\Theta(C)\ge mb-b$.
\end{proof}

\begin{example}
	Consider a primitive $\sigma-$LFSR with $4$, $2$-input $2$-output delay blocks i.e. $m=2$ and $b=4$. Therefore $n =mb= 8$. The primitive polynomial for the companion matrix $P_z$ is $f(x)=x^{8} + x^{4} + x^{3} + x^{2} + 1$. The corresponding matrices $Q$ and $Q_P$ have the following structure:

\begin{equation}
 Q=\left(\begin{array}{rrrrrrrr}
0 & 0 & 0 & 0 & 0 & 0 & 0 & 1 \\
x_{1} & x_{2} & x_{3} & x_{4} & x_{5} & x_{6} & x_{7} & x_{8} \\
0 & 0 & 0 & 0 & 0 & 0 & 1 & 0 \\
x_{2} & x_{3} & x_{4} & x_{5} & x_{6} & x_{7} & x_{8} & x_{1} + x_{3} + x_{4} + x_{5} \\
0 & 0 & 0 & 0 & 0 & 1 & 0 & 0 \\
x_{3} & x_{4} & x_{5} & x_{6} & x_{7} & x_{8} & x_{1} + x_{3} + x_{4} + x_{5} & x_{2} + x_{4} + x_{5} + x_{6} \\
0 & 0 & 0 & 0 & 1 & 0 & 0 & 0 \\
x_{4} & x_{5} & x_{6} & x_{7} & x_{8} & x_{1} + x_{3} + x_{4} + x_{5} & x_{2} + x_{4} + x_{5} + x_{6} & x_{3} + x_{5} + x_{6} + x_{7}
\end{array}\right)
\end{equation}

\begin{equation}
Q^{P}=\begin{array}{c|c}
Q_1=\begin{pmatrix}
0 & 0 & 0 & 0  \\
0 & 0 & 0 & 0  \\
0 & 0 & 0 & 0  \\
0 & 0 & 0 & 0  
\end{pmatrix}
 &
Q_2=\begin{pmatrix} 
0 & 0 & 0 & 1\\
0 & 0 & 1 & 0\\
0 & 1 & 0 & 0\\
1 & 0 & 0 & 0
\end{pmatrix}
\\ \hline
Q_3=\begin{pmatrix}
x_{1} & x_{2} & x_{3} & x_{4}\\
x_{2} & x_{3} & x_{4} & x_{5}\\ 
x_{3} & x_{4} & x_{5} & x_{6}\\ 
x_{4} & x_{5} & x_{6} & x_{7}
\end{pmatrix}
&

Q_4=\begin{pmatrix}
x_{5} & x_{6} & x_{7} & x_{8}\\
x_{6} & x_{7} & x_{8} &x_{1} + x_{3} + x_{4} + x_{5} \\
x_{7} & x_{8} &x_{1} + x_{3} + x_{4} + x_{5} & x_{2} + x_{4} + x_{5} + x_{6} \\
x_{8} & x_{1} + x_{3} + x_{4} + x_{5} & x_{2} + x_{4} + x_{5} + x_{6} & x_{3} + x_{5} + x_{6} + x_{7}
\end{pmatrix}
\end{array}
\end{equation}
The $4$-th row of $Q^{P}$ or the $7$-th column of $Q^{-1}$ is given by
\[Q^{-1}[:,7]=(P_1,P_2,P_3,P_4,1,0,0,0)\] where
\begin{itemize}
    \item[$P_1: $] $x_{2} x_{4} x_{6} x_{8} + x_{2} x_{4} x_{7} + x_{2} x_{5} x_{8} + x_{2} x_{6} + x_{3} x_{6} x_{8} + x_{3} x_{7} + x_{4} x_{5} x_{6} + x_{4} x_{6} + x_{4} x_{8} + x_{5}$.
    \item[$P_2: $] $ x_{1} x_{4} x_{6} x_{8} + x_{1} x_{4} x_{7} + x_{1} x_{5} x_{8} + x_{1} x_{6} + x_{2} x_{3} x_{6} x_{8} + x_{2} x_{3} x_{7} + x_{2} x_{4} x_{5} x_{8} + x_{2} x_{4} x_{6} x_{7} + x_{2} x_{5} x_{6} + x_{2} x_{5} x_{7} + x_{3} x_{4} x_{8} + x_{3} x_{5} x_{6} + x_{3} x_{5} x_{8} + x_{3} x_{6} x_{7} + x_{4} x_{5} + x_{4} x_{7}$
    \item[$P_3: $] $x_{1} x_{3} x_{6} x_{8} + x_{1} x_{3} x_{7} + x_{1} x_{4} x_{5} x_{8} + x_{1} x_{4} x_{6} x_{7} + x_{1} x_{5} x_{6} + x_{1} x_{5} x_{7} + x_{2} x_{3} x_{5} x_{8} + x_{2} x_{3} x_{6} x_{7} + x_{2} x_{4} x_{5} x_{7} + x_{2} x_{4} x_{6} + x_{2} x_{4} x_{8} + x_{2} x_{5} x_{6} + x_{2} x_{6} x_{8} + x_{2} x_{7} + x_{3} x_{4} x_{7} + x_{3} x_{4} x_{8} + x_{3} x_{5} x_{7} + x_{3} x_{5} + x_{4} x_{5} + x_{4} x_{6}$
    \item[$P_4: $] $x_{1} x_{3} x_{5} x_{8} + x_{1} x_{3} x_{6} x_{7} + x_{1} x_{4} x_{5} x_{7} + x_{1} x_{4} x_{6} + x_{1} x_{4} x_{8} + x_{1} x_{5} x_{6} + x_{2} x_{3} x_{5} x_{7} + x_{2} x_{3} x_{6} + x_{2} x_{4} x_{7} + x_{2} x_{5} x_{8} + x_{2} x_{5} + x_{2} x_{6} x_{7} + x_{3} x_{4} x_{6} + x_{3} x_{4} x_{7} + x_{3} x_{8} + x_{4} x_{5}$
\end{itemize}

Here, $C_{\mathcal{S}}[7,7]$, is equal to $P_4$ which is a polynomial of degree 4. 

\end{example}

\section{Case Study: Integration with SNOW 2.0}\label{CaseSt}
In this subsection, we first introduce SNOW 2.0, and then use it as a case study to show how $\sigma$-KDFC can be applied to an LFSR-based cipher stream. We refer to the resulting cipher as KDFC-SNOW.


\subsection{SNOW 2.0:}
The SNOW series of word based stream ciphers was first introduced in \cite{SNOW1-01}. This version of SNOW is known as SNOW 1.0. This was shown to be vulnerable to a linear distinguishing attack as well as a guess and determine attack \cite{SNOW1-02}.

SNOW 2.0 (Adopted by ISO/IEC standard IS 18033-4) was introduced later in \cite{ekdahl2002new} as a modified version of SNOW 1.0. This version was shown to be vulnerable to algebraic and other attacks \cite{billet2005resistance,SNOW1-02,ahmadi2009heuristic,New-SN-07,nia2014new,nyberg2006improved}. We consider SNOW 2.0 as a test case and demonstrate how replacing the LFSR in this scheme with a $\sigma$-LFSR increases its resistance to various attacks.

 The block diagram of SNOW 2.0 is shown in figure \ref{SN2}.

\begin{figure}[H]
	\centering
	\includegraphics[width=0.9\linewidth]{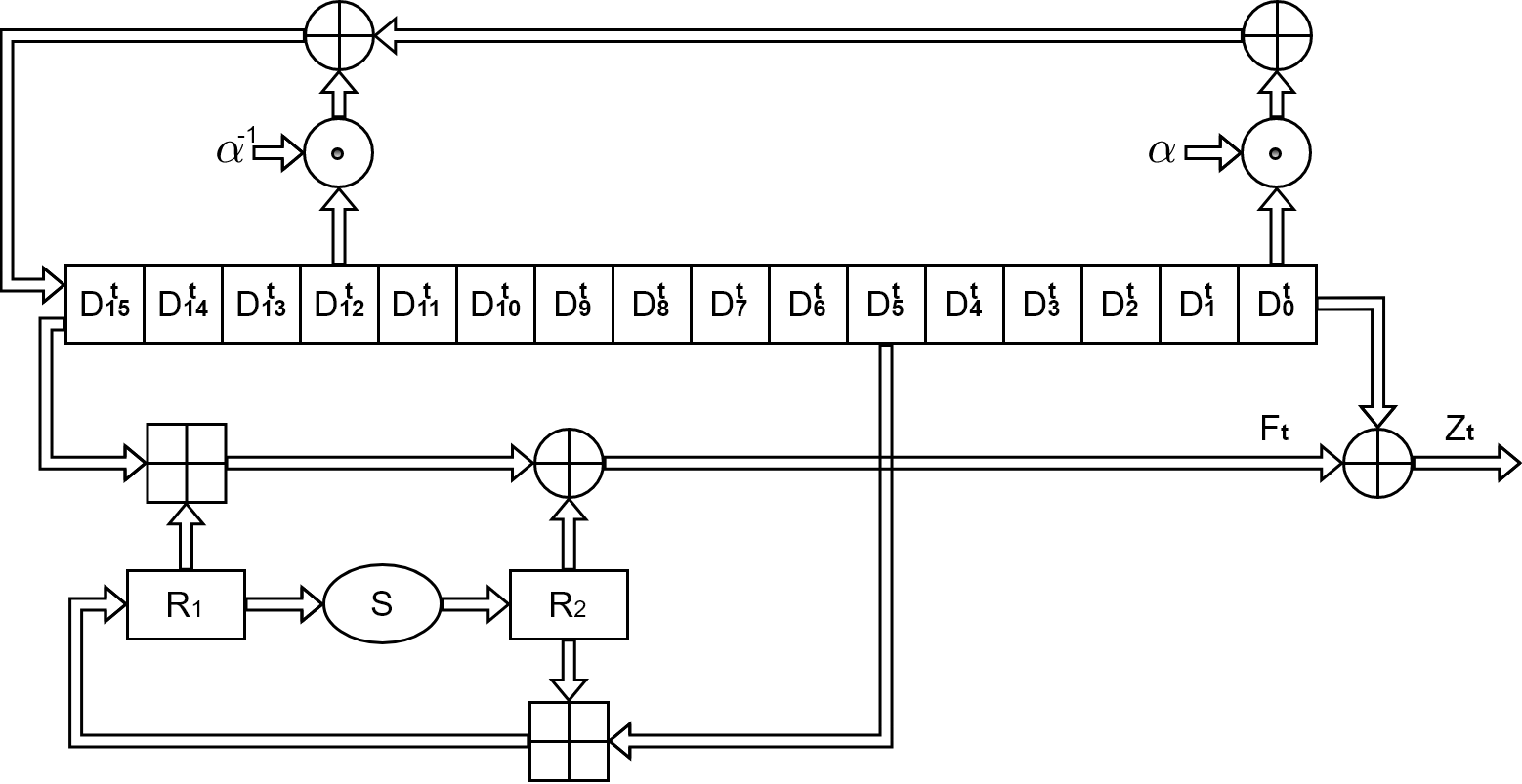}
	\caption{The block diagram of SNOW 2.0}
	\label{SN2}
\end{figure}

In figure \ref{SN2}, $+$ and $\boxplus$ represent bit wise XOR (addition in the field $GF(2)$)  and integer addition modulo $2^{32}$ respectively.  As shown in figure \ref{SN2}, the keystream generator in SNOW 2.0 consists of an LFSR and an FSM (Feedback State Machine). The LFSR implements the following linear recurring relation:
\begin{equation*}
  x_{n+16}=\alpha^{-1}x_{n+11}+x_{n+2}+\alpha x_{0}. 
\end{equation*}

 where, $\alpha$ is the root of the following primitive polynomial
 \begin{equation*}
   G_{S}(x)=(x^4+\beta^{23}x^{3}+\beta^{245}x^{2}+\beta^{48}x+\beta^{239})   \in \mathbb{F}_{2^{8}}[X]  
 \end{equation*}
 where $\beta$ is the root of the following primitive polynomial.
 
 \begin{equation*}
   H_{S}(x)=x^8+x^7+x^5+x^3+1  \in \mathbb{F}_{2}[X]  
 \end{equation*}

The FSM contains two 32-bit registers $R1$ and $R2$. These registers are connected by means of an S-Box which is made using four AES S-boxes. This S-box serves as the source of nonlinearity. 



\subsection{KDFC-SNOW:}
 In the proposed modification, we replace the LFSR part of SNOW 2.0 by a $\sigma$-LFSR having 16, 32-input 32-output delay blocks. The configuration matrix of the $\sigma$-LFSR is generated using Algorithm 1. We shall refer to the modified scheme, shown in Figure \ref{arch}, as KDFC-SNOW.

\begin{figure}[H]
	\centering
	\includegraphics[width=1\linewidth]{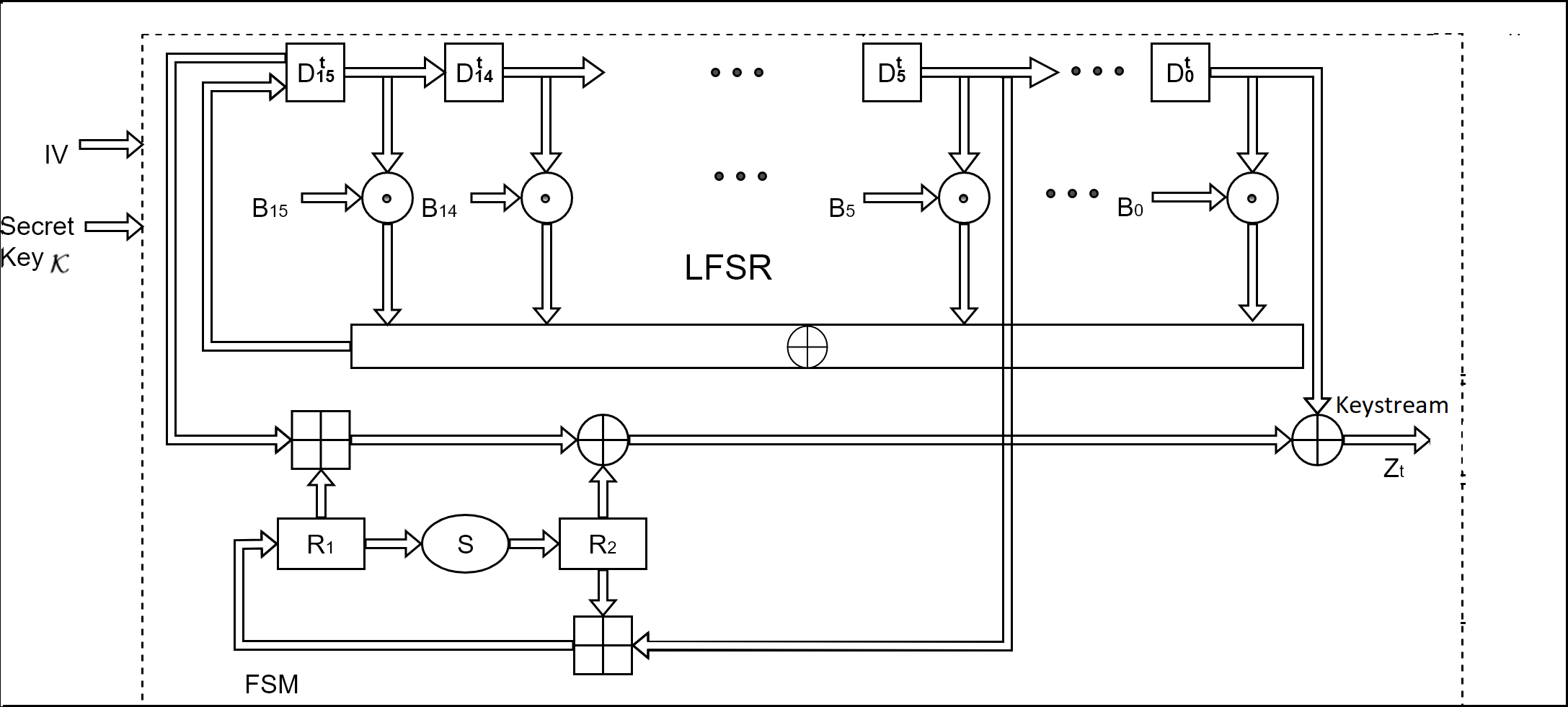}
	\caption{The block diagram KDFC-SNOW}
	\label{arch}
\end{figure}

During initialization the feedback function of the $\sigma$-LFSR is identical to that of SNOW-2.  As in SNOW 2.0, the $\sigma$-LFSR is initialized using a $128$-bit IV and a $128/256$-bit secret key ${K}$. KDFC-SNOW is run with this configuration for 32 clock cycles without producing any symbols at the output. The vectors generated in the last 12 of these clock cycles are used in Algorithms 1 to generate a new feedback configuration. As we have already mentioned, some of the iterations of Algorithm 1 are pre-calculated and the remaining ones are done as a part of the initialization process. In this case, it is assumed that 468 of these iterations are pre-calculated and the last 12 iterations are carried out in the initialization process. This calculated configuration replaces the original one and the resulting set-up is used to generate the keystream. 

\subsection{Initialization of KDFC-SNOW}

\begin{itemize}
    \item The delay blocks $D_0,\cdots, D_{15}$ are initialized using the 128/256 bit secret key $K$ and a 128 bit  $IV$ in exactly the same manner as SNOW 2.0. The registers $R_1$ and $R_2$ are set to zero.

   \item The initial feedback configuration of the $\sigma$-LFSR is identical to SNOW 2.0. This is done by setting $B_{11}$ and $B_0$ as  matrices that represent multiplication by $\alpha^{-1}$ and $\alpha$ respectively.  Further,  $B_2$ is set to identity. The other gain matrices are set to zero.
    \item KDFC-SNOW is run in this configuration for 32 clock cycles without making the output externally available. 
    The last 12 values of $F^t$ are used as the random numbers in Algorithm 1.
    \item A new configuration matrix is calculated using Algorithms 1 and the corresponding feedback configuration replaces the original one. The scheme is now run with this configuration. The first 32 vectors are discarded and the key stream starts from the $33^{rd}$ vector.
\end{itemize}


\vspace{1cm}


\subsection{Governing Equations of KDFC-SNOW}
Let $D_i^t \in \mathbb{F}_{{2}^{32}}$ denote the value stored in the $i^{th}$ delay block at the $t$-th time instant after the key stream generation has started.
The outputs of the delay blocks of the $\sigma$-LFSR are related as per the following equation:
\begin{equation}
 D_{15}^{t+1}=B_{0}D_{0}^{t} + B_{1}D_{1}^t + \cdots + B_{15}D_{15}^t   
\end{equation}
Therefore,
\begin{equation}\label{Up}
D_{k}^{t+1}=\begin{cases}
D_{k+t+1}^{0} &                                              0 \leq k+t+1 \leq 15\\
B_{0}D_{0}^{t+k} + B_{1}D_{1}^{k+t} + \cdots + B_{15}D_{15}^{t+k} & k+t+1 >15
\end{cases}
\end{equation}
The value of the keystream at the $t$-th time instant is given by the following equation
 Let $F_t$ be the output of the FSM at time $t$,
\begin{equation} \label{FSMEq1}
    F_t=(D_{15}^t \boxplus R_1^t ) + R_2^t 
\end{equation}
The registers are updated as follows:
\begin{equation}\label{FSMEq2}
        R_1^{t+1}=D_5^t \boxplus R_2^t
 \end{equation} 
 \begin{equation}\label{FSMEq3}
       R_2^{t+1}=S(R_1^t)   
 \end{equation}
\begin{equation}\label{KS}
z^{t}=R_1^{t}\boxplus D_{15}^{t} + R_2^{t} + D_{0}^{t}
=(R_2^{t-1} \boxplus D_{4}^{t}) \boxplus D_{15}^t +R_2^{t}+ D_{0}^{t}
\end{equation}
where $R_1^t$ and $R_2^t$ represent the values of registers $R_1$ and $R_2$ at time instant $t$. The operation "$\boxplus$" is defined as follows:

\begin{equation}\label{bx}
   x \boxplus y=(x+y) \mod 2^{32}
\end{equation}

The challenge for an adversary in this scheme is to find the gain matrices \{$B_0,B_1,\cdots,B_{15}\}$ in addition to the initial state $\{D_{0}^{0},\cdots,D_{15}^{0}\}$.

Note that Equations \ref{FSMEq1} to \ref{KS} are got from the FSM. Since the FSM part of the keystream generator is identical for SNOW 2.0 and KDFC-SNOW these equations are identical for both schemes.

\subsection{Security enhancement due to KDFC-SNOW}\label{Eval}
\subsubsection{Algebraic Attack: }
We first briefly the Algebraic attack on SNOW 2 described in \cite{billet2005resistance} and demonstrate why this attack becomes difficult with KDFC-SNOW. This attack first attempts to break a modified version of the scheme where the $\boxplus$ operator is approximated by  $\oplus$. The state of LFSR and the  value of the registers at the end of the 32 initialization cycles are considered unknown variables.  This accounts for a total of  $512+32=544$ unknown variables. The algebraic degree of each of the S-box(S) equations (156 linearly independent quadratic equations in each clock cycle ) is 2.
Rearranging the terms in Equation \ref{KS}, we get the following

\begin{equation}\label{RRE}
 R_2^{t}=(R_2^{t-1} \boxplus D_{4}^{t})  \boxplus D_{15}^t+  D_{0}^{t} + z^{t}.
\end{equation}

Note that $R_1^{0}=R_2^{0}+z^{0}+D_{0}^{0}+D_{15}^{0}$.Therefore, by approximating $\boxplus$ as $\oplus$, Equation \ref{RRE} expands to the following:

\begin{equation}\label{Exp}
R_2^{t}=R_2^0+\sum_{i=0}^{t} z^{i}+ \sum_{i=0}^{t}(D_{4}^{i}+D_{15}^{i}+D_{0}^{i})
\end{equation}

Further, Equation \ref{FSMEq3} can be expanded as follows:
 
\begin{equation}\label{SBB}
R_2^{t+1}=S(R_1^t)=S(R_2^t+z^t+D_{15}^t+D_{0}^t)
\end{equation}
In Equation \ref{SBB}, the outputs of the delay blocks  can be related to the initial state of the LFSR using  the following equation.
\begin{equation}\label{Up1}
D_{k}^{t+1}=\begin{cases}
D_{k+t+1}^{0} &                                              0 \leq k+t+1 \leq 15\\
\alpha^{-1}D_{11}^{k+t}+D_{2}^{k+t}+\alpha D_{0}^{k+t} & k+t+1 >15
\end{cases}
\end{equation}
Because of the nature of the $S$-Box, Equation \ref{SBB} gives rise to 156 quadratic equations per time instant (\cite{billet2005resistance}).
 When these equations are linearized, the number of variables increases to $\sum_{i=0}^{2}\binom{544}{i} \approx 2^{17}$. Therefore with  $2^{17}/156 \approx 951$ samples, we get a system of equations, which can be solved in $\mathcal{O}(2^{51})$ time, to obtain the initial state of the LFSR and the registers. This attack is then modified to consider the $\boxplus$ operator. This attack has a time complexity of  approximately $\mathcal{O}(2^{294})$.
 
 When the LFSR in SNOW 2.0 is replaced by a $\sigma$-LFSR, the feedback equation is no longer known.  If the entries of the feedback gain matrices are considered as unknowns, then there are a total of  $16*m^2+mb + m = 16928 $ unknown variables (This includes the $16*m^2$ entries of the feedback matrices and $mb+m$ entries corresponding to the state of the LFSR and the register $R_2$ at the beginning of the key stream). The output of the delay blocks at any given instant are  functions of these variables. Here, Equation \ref{Up1} is replaced by Equation \ref{Up}. Now, if the outputs of the delay blocks in Equation \ref{SBB} are linked to the initial state of the LFSR (i.e. the state when the key stream begins) using Equation \ref{Up} instead of Equation \ref{Up1}, then the resulting equations contain the feedback matrices and their products. For example, in the expressions for $R_2^2$ and $R_2^3$, $D_{15}^1$ and $D_{15}^2$ are given as follows
 \begin{eqnarray*}
 D_{15}^1 &=& B_0D_0^0 + B_1D_1^0 + \hdots +B_{15}D_{15}^0\\
 D_{15}^2 &=& B_0D_0^1 + B_1D_1^1 + \hdots +B_{15}D_{15}^1\\
 &=& B_0D_1^0 + B_1D_2^0 + \hdots +B_{15}(B_0D_0^0 + B_1D_1^0 + \hdots +B_{15}D_{15}^0)\\
 &=& B_{15}B_0D_0^0 + (B_{15}B_1 +B_0)D_1^0 + \hdots + (B_{15}B_{15} + B_{14})D_{15}^0
 \end{eqnarray*}
 Observe that, while $D_{15}^1$ is a polynomial of degree two in the unknown variables, $D_{15}^2$ is a polynomial of degree 3. Similarly, with each successive iteration the degree of the expression for $D_{15}^t$ keeps increasing, till all the $m^2b$ entries of the feedback matrices are multiplied with each other. A similar thing happens with the expressions for $D_0^t$. This results in a set of polynomial equations having maximum degree equal to $m^2b+1 = 16385$  .  
 Therefore, although the equations generated by  Equation \ref{SBB} are quadratic in terms of the initial state of the $\sigma$-LFSR, they are no longer quadratic in the set of all unknowns. We instead have a system of equations in $16982$ variables with a maximum degree of degree of over $16000$. Linearizing such a system will give us a system of linear equations in $N = \sum_{i=0}^{M}{16416 \choose i}$ unknowns where $M$ is higher than 16000. Such an attack is therefore not feasible.
 
 One could instead consider the rows of the matrix $Y$ generated by Algorithm 1 as unknowns. Assuming that the first row is $e_1^n$, the total number of unknowns will now be $31*512 = 15872$. As we have already seen, the entries of the feedback matrices ($B_i$s) are polynomials in these variables. From Theorem \ref{IC}, the maximum degree of these polynomials is atleast $mb-b$. Therefore, the maximum degree of the equations generated by Equation \ref{SBB} will be atleast $mb-b+1 = 497$. Therefore, linearizing this system of equations  gives rise to a system of linear equations in $N = \sum_{i=0}^{497}{16416 \choose i} \approx \mathcal{O}(2^{3207})$ unknowns. Therefore, an algebraic attack on this scheme that uses linearization seems unfeasible.

\subsubsection{Distinguishing Attack: }
In the distinguishing attack, the attacker aims to distinguish the generated keystream from a random sequence. Distinguishing attacks on SNOW 2.0 have been launched using the linear masking method  \cite{watanabe2003distinguishing, nyberg2006improved, lee2008cryptanalysis}. This method essentially adapts the linear cryptoanalysis method given in \cite{matsui1993linear} to stream ciphers. In this method, the algorithm of the key stream generator is assumed to consist of two parts, a linear one and a non-linear one. In the case of SNOW 2.0, the linear part is the LFSR and the non-linear part is the FSM. The linear part satisfies a linear recurring relation of the form $f(x_n,x_{n+1},x_{n+2},\ldots,x_{n+k})=0$ for all $n$. We then try to find a linear relation, called the masking relation, that the non-linear part approximately satisfies. This relation is of the following form:
\begin{equation}\label{Mask}
    \sum_{i=0}^{\ell_1}\Gamma_ix_{n+i} = \sum_{i=0}^{\ell_2}\Lambda_iz_{n+i}
\end{equation}
where $z_0,z_1,\ldots$ is the output sequence of the key stream generator. The $\Gamma_i$s and $\Lambda$s are linear masks that map the corresponding $x_{n+i}$s and $z_{n+i}$s to $\mathbb{F}_2$ respectively. The error in the masking relation can be seen as a random variable. If $p$ is the probability that the non linear part satisfies Equation \ref{Mask}, then $p-\frac{1}{2}$ is called the bias of the masking relation. The Masking relation along with the linear recurring relation is used to generate a relation in terms of the elements of the output sequence. The error in this relation can also be seen as a random variable. If the probability of the sequence satisfying this relation is $p_f$, then $p_f - \frac{1}{2}$ is the bias of this relation. This bias can be related to the bias of the masking relation using the piling up lemma in \cite{matsui1993linear}. The main task in this type of attack is to find masks $\Gamma_i$s and $\Gamma_i^{\prime}$s which maximise the bias of the masking relation.
The following linear masking equation is used in \cite{watanabe2003distinguishing} and \cite{nyberg2006improved} for the FSM of SNOW 2.
\begin{equation} \label{MaskSNOW}
    \Gamma_0x_n +\Gamma_1x_{n+1} +\Gamma_5x_{n+5} + \Gamma_{15}x_{n+15} + \Gamma_{16}x_{n+16} = \Lambda_0z_{n} +\Lambda_1z_{n+1}
\end{equation}
In \cite{watanabe2003distinguishing} it is assumed that all the $\Gamma_i$s and $\Lambda_i$s are equal to each other. In \cite{nyberg2006improved} it is assumed that $\Gamma_0,\Gamma_{15}$ and $\Lambda_0$ are equal to each other.  $\Gamma_1,\Gamma_5.\Gamma_{16}$ and $\Lambda_1$ are also assumed to be equal. Since $f(x_n,x_{n+1},\ldots,x_{n+k})$ is a linear relation, the following relation can be written purely in terms of the $z_i$s
\begin{eqnarray*}
 \Gamma_0f(x_n,x_{n+1},\ldots,x_{n+k}) &+&\Gamma_1f(x_{n+1},x_{n+2},\ldots,x_{n+k+1})\\ +\Gamma_5f(x_{n+5},x_{n+6},\ldots,x_{n+k+5}) &+& \Gamma_{15}f(x_{n+15},x_{n+16},\ldots,x_{n+k+15})\\ + \Gamma_{16}f(x_{n+16},x_{n+17},\ldots,x_{n+k+16})&=&0
\end{eqnarray*}
The linear relation between the elements of the output sequence in both \cite{watanabe2003distinguishing} and \cite{nyberg2006improved} is obtained using this method. Further, if there are $\ell$ non-zero coefficients in $f$, then the random variable corresponding to the error in this relation is a  sum of $\ell$ random variables each corresponding to the error in the linear masking equation.

In the proposed $\sigma$-LFSR configuration, the feedback equation is not known. Therefore the only known linear recurring relation that the output of the $\sigma$-LFSR satisfies is the one defined by its characteristic polynomial. If the characteristic polynomial is assumed to be the same as that of the LFSR in SNOW 2, then the corresponding linear recurring relation has 250 non-zero coefficients. Further, since these coefficients are elements of $\mathbb{F}_2$, the non-zero coefficients are all equal to 1. Therefore, as a consequence of the piling up lemma, if the bias of the masking equation is $\epsilon$, then the bias of the relation between the elements of the key stream is given as follows
\begin{equation}
    \epsilon_{final} = 2^{249} \times \epsilon^{250}
\end{equation}
  The number of elements of the key stream needed to distinguish it from a random sequence is $\frac{1}{\epsilon_{final}^2}$. Therefore, for an identical linear masking equation, the length of the key stream for the distinguishing attack is much higher for the proposed configuration as compared to SNOW 2. This is demonstrated in the following table.
  
  	\begin{table}[H]
		\centering
		\scalebox{.75}{\begin{tabular}{|c|c|c|c|c|c|}
				\hline
				\textbf{Reference}  & \textbf{$\epsilon$} & \textbf{$\epsilon_{final}$} & \textbf{$\epsilon_{final}$} & \textbf{\#Keystream} &
				\textbf{\#Keystream}\\
				\textbf{}  &  & \textbf{ (SNOW 2.0)} & \textbf{ (KDFC-SNOW)} & \textbf{(SNOW 2.0)} &
				\textbf{(KDFC-SNOW)}\\
				\hline
				\cite{watanabe2003distinguishing}  & $2^{-27.61}$ & $2^{-112.25}$ & $2^{-6653.5}$ & $2^{225}$ & $2^{13307}$\\
				\hline
				\cite{nyberg2006improved} &  $2^{-15.496}$ & $2^{-86.9}$ & $2^{-3625}$ & $2^{174}$ & $2^{7250}$\\
				\hline
		\end{tabular}}
		\caption{Comparison of required length of key stream for the distinguishing attack }
	\end{table}

\subsubsection{Fast Correlation Attack: }
The Fast Correlation Attack is a commonly used technique for the cryptanalysis of LFSR based stream ciphers. This method was first introduced for bitwise keystreams in (\cite{meier1989fast}). Here, the attacker views windows of the key stream as noisy linear encodings of the initial state of the LFSR. She then tries to recover the initial state by decoding this window. Further, linear combinations of elements in this window can be seen as encodings of subsets of the initial state. This results in smaller codes which are more efficient to decode \cite{chepyzhov2000simple}. The linear recurring relation satisfied by the output of the LFSR is used to generate the parity check matrix for this code.  A Fast correlation attack for word based stream ciphers was first described in \cite{jonsson2001correlation}. An improvement on this attack is given in \cite{lee2008cryptanalysis}.  Both these schemes consider a linear recurring relation with coefficients in $\mathbb{F}_2$. For SNOW 2.0, this relation has order 512. This is equivalent to considering each component sequence to be generated by a conventional bitwise LFSR having the same characteristic polynomial as the LFSR in SNOW 2.0.  The time complexity of the attack given in \cite{lee2008cryptanalysis} is $2^{212.38}$. The scheme in \cite{zhang2015fast} considers the LFSR in SNOW 2.0 to be over $\mathbb{F}_{2^8}$. This results in a linear recurring relation of order 64. This modification results in a significant improvement in the time complexity of the attack. The time complexity of this attack is $2^{164.5}$ which is around $2^{49}$ times better than that of the attack given in  \cite{lee2008cryptanalysis}. However, in order to derive the linear recurring relation over $\mathbb{F}_{2^8}$, the knowledge of the feedback function is critical.

In KDFC SNOW, the characteristic polynomial of the $\sigma$-LFSR is publicly known. The attacker can therefore generate a linear recurring relation over $ \mathbb{F}_2$ that the output of the $\sigma$-LFSR satisfies. Therefore, the attack given in \cite{lee2008cryptanalysis} will also be effective against KDFC-SNOW. However, without the knowledge of the feedback function, the attacker will not be able to derive a linear recurring relation over $\mathbb{F}_{2^8}$. Hence, KDFC-SNOW is resistant against the attack given in \cite{zhang2015fast}. Thus.  the best time complexity that a known fast correlation attack can achieve against KDFC SNOW is $2^{212.38}$

\subsection{Guess and Determine Attack}
In a Guess and Determine Attack, the attacker aims to estimate the values of a minimum number of variables using which the complete sequence can be constructed. For SNOW 2.0, this includes the values of the outputs of the delay blocks of the LFSR and the outputs of the registers of the FSM at some time instant. This is done by guessing some of the values and determining the rest of them using system equations. If the sequence generated using these estimates matches the output of the key-stream generator, then the guesses are deemed to be correct. Otherwise, a fresh set of guesses are considered. The set of variables whose values are guessed is known as the basis for the attack. For both SNOW 2.0 and KDFC-SNOW, these variables take their values from $\mathbb{F}_2^{32}$. Hence, if the size of the basis is $k$, then the probability of a correct guess is $2^{-32k}$. Thus, on average, one needs $\mathcal{O}(2^{-32k})$ attempts to make a correct guess. Therefore, the problem here is to find a basis of the minimum possible size. A systematic Vitterbi-like algorithm for doing this is given in \cite{ahmadi2009heuristic}.The complexity of this attack was found to be $2^{265}$(\cite{ahmadi2009heuristic}) for SNOW2.0.  The complexity of this attack reduced to $2^{192}$ in (\cite{nia2014new}) by incorporating a couple of auxiliary equations. We now briefly describe the algorithm given in \cite{ahmadi2009heuristic} in the context of SNOW 2.0

Consider the following equations which are satisfied by   SNOW 2.0
\begin{eqnarray}
 D_t^{16}&=&\alpha ^{-1} D_t^{11} + D_2^{t} + \alpha D_t \label{E1}\\    R1_t&=&D_t^4 + S(R1_{t-2}) \label{E2} \\
 z_t&=&D_t + (D_t^{15} + R1_t) + S(R1_{t-1}) \label{E3}
\end{eqnarray}
 These equations are used to generate the following tables
\begin{table}

	\begin{tabular}{|c |c | c |c|}

			\hline
			0 & 2 & 11 & 16\\
			\hline
			1 & 3 &12 & 17\\
			\hline
			$\vdots$ & $\vdots$ & $\vdots$  & $\vdots$ \\
			\hline
			18 & 20 & 29 &34 \\
			\hline
			\end{tabular}
			\hfill
				\begin{tabular}{|c |c |c |}

			\hline
			4 & 35 & 37\\
			\hline
			5 & 36  & 38\\
			\hline
			$\vdots$ & $\vdots$  & $\vdots$\\
			\hline
			22 & 53 & 55\\
			\hline
			\end{tabular}
			\hfill
				\begin{tabular}{|c |c |c |c |c | c |}

			\hline
			0 & 15 & 36 & 37 \\
			\hline
			1 & 16 & 37 & 38 \\
			\hline
			$\vdots$ & $\vdots$ & $\vdots$  & $\vdots$ \\
			\hline
			18 & 33 & 54 & 55\\
			\hline
			\end{tabular}\\
			\caption{Index table for SNOW 2.0}
	\end{table}
The entries in the above tables correspond to the variables that are to be estimated. The entries in the first table, i.e. 0 to 34, correspond to 35 consecutive outputs of the LFSR. The entries 35 to 55 correspond to 21 consecutive entries of Register $R1$. Each row of the above tables correspond to the values of the delay blocks and registers in Equations  \ref{E1}, \ref{E2} and \ref{E3} at a particular time instant. 

We now consider a multi stage graph with 56 nodes in each stage corresponding to the 56 entries in the above tables. Each node is connected to all the nodes in the next stage giving rise to a trellis diagram. An entry is said to be eliminated by a path if, knowing the values of the entries corresponding to the nodes in the path, the value of that entry can be calculated. 

We now recursively calculate an optimal path that eliminates all the entries. The desired basis corresponds to the nodes in this path. In the $i$-th iteration of this algorithm we calculate the optimal path of length $i$ to each node in the $i$-th stage. In order to find the optimal path to the $k$-th node, we consider all the incoming edges of node $k$. By appending the node $k$ to the optimal paths of length $i-1$ ending at the source nodes of these edges, we get 55 paths of length $i$. We choose the edge corresponding to the path that eliminates the most number of variables. In case of a tie, we consider the path that results in the most number of rows with 2 unknowns and so on. This process is continued till we get a path that eliminates all the entries. This algorithm results in a basis of cardinality 8 for SNOW 2.0.

In KDFC-SNOW, the feedback equation of the $\sigma$-LFSR is not known. The smallest known linear recurring relation that the output of the $\sigma$-LFSR satisfies is the relation corresponding to its characteristic polynomial. This relation is given as follows
\begin{eqnarray*} 
    x_{n+512} = x_{n+510} + x_{n+504} + x_{n+502} + x_{n+501} + x_{n+494} + x_{n+493} + x_{n+490} +\\ x_{n+486} + x_{n+485} + x_{n+483} + x_{n+481} + x_{n+480} + x_{n+478} + x_{n+477} + x_{n+471} +\\ x_{n+470} + x_{n+469} + x_{n+466} + x_{n+462} + x_{n+461} + x_{n+459} + x_{n+458} + x_{n+452} +\\ x_{n+449} + x_{n+446} + x_{n+445} + x_{n+444} + x_{n+441} + x_{n+438} + x_{n+437} + x_{n+434} +\\ x_{n+433} + x_{n+432} + x_{n+431} + x_{n+429} + x_{n+427} + x_{n+424} + x_{n+423} + x_{n+420} +\\ x_{n+419} + x_{n+414} + x_{n+412} + x_{n+411} + x_{n+409} + x_{n+405} + x_{n+402} + x_{n+400} +\\ x_{n+399} + x_{n+398} + x_{n+396} + x_{n+395} + x_{n+393} + x_{n+392} + x_{n+390} + x_{n+388} +\\ x_{n+387} + x_{n+385} + x_{n+375} + x_{n+374} + x_{n+372} + x_{n+371} + x_{n+366} + x_{n+365} +\\ x_{n+363} + x_{n+362} + x_{n+359} + x_{n+357} + x_{n+356} + x_{n+355} + x_{n+354} + x_{n+353} +\\ x_{n+352} + x_{n+351} + x_{n+350} + x_{n+347} + x_{n+345} + x_{n+344} + x_{n+343} + x_{n+341} +\\ x_{n+339} + x_{n+338} + x_{n+337} + x_{n+336} + x_{n+333} + x_{n+330} + x_{n+329} + x_{n+326} +\\ x_{n+324} + x_{n+322} + x_{n+319} + x_{n+310} + x_{n+307} + x_{n+306} + x_{n+305} + x_{n+304} +\\ x_{n+303} + x_{n+301} + x_{n+299} + x_{n+298} + x_{n+297} + x_{n+296} + x_{n+295} + x_{n+294} +\\ x_{n+293} + x_{n+292} + x_{n+291} + x_{n+289} + x_{n+286} + x_{n+285} + x_{n+283} + x_{n+282} +\\ x_{n+281} + x_{n+278} + x_{n+276} + x_{n+274} + x_{n+271} + x_{n+269} + x_{n+264} + x_{n+262} +\\ x_{n+259} + x_{n+258} + x_{n+257} + x_{n+255} + x_{n+253} + x_{n+251} + x_{n+249} + x_{n+248} +\\ x_{n+243} + x_{n+240} + x_{n+239} + x_{n+238} + x_{n+236} + x_{n+235} + x_{n+233} + x_{n+232} +\\ x_{n+230} + x_{n+229} + x_{n+228} + x_{n+227} + x_{n+226} + x_{n+222} + x_{n+217} + x_{n+216} +\\ x_{n+215} + x_{n+214} + x_{n+213} + x_{n+210} + x_{n+208} + x_{n+206} + x_{n+203} + x_{n+201} +\\ x_{n+199} + x_{n+193} + x_{n+190} + x_{n+184} + x_{n+179} + x_{n+178} + x_{n+177} + x_{n+175} +\\ x_{n+174} + x_{n+173} + x_{n+172} + x_{n+171} + x_{n+169} + x_{n+165} + x_{n+164} + x_{n+163} +\\ x_{n+158} + x_{n+156} + x_{n+155} + x_{n+153} + x_{n+152} + x_{n+151} + x_{n+149} + x_{n+147} +\\ x_{n+146} + x_{n+143} + x_{n+141} + x_{n+138} + x_{n+136} + x_{n+132} + x_{n+131} + x_{n+129} +\\ x_{n+128} + x_{n+126} + x_{n+125} + x_{n+124} + x_{n+123} + x_{n+121} + x_{n+120} + x_{n+119} +\\ x_{n+118} + x_{n+117} + x_{n+116} + x_{n+115} + x_{n+113} + x_{n+112} + x_{n+111} + x_{n+109} +\\ x_{n+105} + x_{n+104} + x_{n+103} + x_{n+102} + x_{n+98} + x_{n+97} + x_{n+94} + x_{n+93} + x_{n+89}\\ + x_{n+88} + x_{n+87} + x_{n+81} + x_{n+78} + x_{n+76} + x_{n+75} + x_{n+73}\\ + x_{n+72} + x_{n+70} + x_{n+69} + x_{n+68} + x_{n+67} + x_{n+66} + x_{n+65} + x_{n+63}\\ + x_{n+59} + x_{n+58} + x_{n+57} + x_{n+56} + x_{n+55} + x_{n+53} + x_{n+51} + x_{n+50}\\ + x_{n+49} + x_{n+47} + x_{n+46} + x_{n+45} + x_{n+44} + x_{n+41} + x_{n+39} + x_{n+37}\\ + x_{n+36} + x_{n+33} + x_{n+30} + x_{n+26} + x_{n+25} + x_{n+21} + x_{n+20} + x_{n+19}\\ + x_{n+16} + x_{n+5} + x_0
\end{eqnarray*}
As in SNOW 2.0, the following equations are also satisfied,
\begin{eqnarray*}
    R1_n&=&x_{n+4} + S(R1_{n-1})\\
 z_n&=&x_n + (x_{n+15} + R1_n) + S(R1_{n-1})
\end{eqnarray*}
The following tables can be constructed using these three equations.
 \begin{table}

	\begin{tabular}{|c |c |c |c |c | c |c|}

			\hline
			512 & 510 & 504 & 502 & $\cdots$ & 5 & 0\\
			\hline
			513 & 511 &505 & 503&$\cdots$ & 6 & 1\\
			\hline
			$\vdots$ & $\vdots$ & $\vdots$  & $\vdots$ & $\vdots$& $\vdots$ & $\vdots$\\
			\hline
			1025 & 1023 & 1007&1005 &$\cdots$ & 516 & 513\\
			\hline
			\end{tabular}
			\hfill
				\begin{tabular}{|c |c |c |}

			\hline
			4 & 1026 & 1028\\
			\hline
			5 & 1027  & 1029\\
			\hline
			$\vdots$ & $\vdots$  & $\vdots$\\
			\hline
			517 & 1539 & 1541\\
			\hline
			\end{tabular}
			\hfill
				\begin{tabular}{|c |c |c |c |c | c |}

			\hline
			0 & 15 & 1027 & 1028 \\
			\hline
			1 & 16 & 1028 & 1029 \\
			\hline
			$\vdots$ & $\vdots$ & $\vdots$  & $\vdots$ \\
			\hline
			513 & 528 & 1540 & 1541\\
			\hline
			\end{tabular}\\
			\caption{Index table for $f1(x)$,equation $32$ and equation $34$}
	
	\end{table}
	We ran the Vitterbi-like algorithm with the above tables on a cluster with 40 INTEL(R) XEON(R) CPUs (E5-2630 2.2GHz). The program ran for 16 iterations and generated the path \{1041, 17, 15, 13, 28, 16, 11, 14, 9, 1050, 18, 39, 12, 7, 5, 0, 3\}. This path has length 17.  This corresponds to a time complexity of  $\mathcal{O}(2^{544})$.

\subsection{Randomness Test}
In this subsection, we evaluate the randomness of the keystream generated by KDFC-SNOW.

\subsubsection{Test  Methodology\\}

We have used the NIST randomness test suite to evaluate the randomness of a keystream generated by KDFC-SNOW.There are are $16$ randomness tests in the suite. Each test returns a level of significance i.e. $P-Value$. If this value is above $0.01$ for a given test, then the keystream is considered to be random for that test.

KDFC-SNOW has been implemented using SageMath 8.0. The NIST randomness tests have been conducted on the generated keystream using  Python 3.6. The characterestic polynomial of the $\sigma$-LFSR has been taken as $f(x)=x^{512} + x^{510} + x^{504} + x^{502} + x^{501} + x^{494} + x^{493} + x^{490} + x^{486} + x^{485} + x^{483} + x^{481} + x^{480} + x^{478} + x^{477} + x^{471} + x^{470} + x^{469} + x^{466} + x^{462} + x^{461} + x^{459} + x^{458} + x^{452} + x^{449} + x^{446} + x^{445} + x^{444} + x^{441} + x^{438} + x^{437} + x^{434} + x^{433} + x^{432} + x^{431} + x^{429} + x^{427} + x^{424} + x^{423} + x^{420} + x^{419} + x^{414} + x^{412} + x^{411} + x^{409} + x^{405} + x^{402} + x^{400} + x^{399} + x^{398} + x^{396} + x^{395} + x^{393} + x^{392} + x^{390} + x^{388} + x^{387} + x^{385} + x^{375} + x^{374} + x^{372} + x^{371} + x^{366} + x^{365} + x^{363} + x^{362} + x^{359} + x^{357} + x^{356} + x^{355} + x^{354} + x^{353} + x^{352} + x^{351} + x^{350} + x^{347} + x^{345} + x^{344} + x^{343} + x^{341} + x^{339} + x^{338} + x^{337} + x^{336} + x^{333} + x^{330} + x^{329} + x^{326} + x^{324} + x^{322} + x^{319} + x^{310} + x^{307} + x^{306} + x^{305} + x^{304} + x^{303} + x^{301} + x^{299} + x^{298} + x^{297} + x^{296} + x^{295} + x^{294} + x^{293} + x^{292} + x^{291} + x^{289} + x^{286} + x^{285} + x^{283} + x^{282} + x^{281} + x^{278} + x^{276} + x^{274} + x^{271} + x^{269} + x^{264} + x^{262} + x^{259} + x^{258} + x^{257} + x^{255} + x^{253} + x^{251} + x^{249} + x^{248} + x^{243} + x^{240} + x^{239} + x^{238} + x^{236} + x^{235} + x^{233} + x^{232} + x^{230} + x^{229} + x^{228} + x^{227} + x^{226} + x^{222} + x^{217} + x^{216} + x^{215} + x^{214} + x^{213} + x^{210} + x^{208} + x^{206} + x^{203} + x^{201} + x^{199} + x^{193} + x^{190} + x^{184} + x^{179} + x^{178} + x^{177} + x^{175} + x^{174} + x^{173} + x^{172} + x^{171} + x^{169} + x^{165} + x^{164} + x^{163} + x^{158} + x^{156} + x^{155} + x^{153} + x^{152} + x^{151} + x^{149} + x^{147} + x^{146} + x^{143} + x^{141} + x^{138} + x^{136} + x^{132} + x^{131} + x^{129} + x^{128} + x^{126} + x^{125} + x^{124} + x^{123} + x^{121} + x^{120} + x^{119} + x^{118} + x^{117} + x^{116} + x^{115} + x^{113} + x^{112} + x^{111} + x^{109} + x^{105} + x^{104} + x^{103} + x^{102} + x^{98} + x^{97} + x^{94} + x^{93} + x^{89} + x^{88} + x^{87} + x^{81} + x^{78} + x^{76} + x^{75} + x^{73} + x^{72} + x^{70} + x^{69} + x^{68} + x^{67} + x^{66} + x^{65} + x^{63} + x^{59} + x^{58} + x^{57} + x^{56} + x^{55} + x^{53} + x^{51} + x^{50} + x^{49} + x^{47} + x^{46} + x^{45} + x^{44} + x^{41} + x^{39} + x^{37} + x^{36} + x^{33} + x^{30} + x^{26} + x^{25} + x^{21} + x^{20} + x^{19} + x^{16} + x^{5} + 1$. 

(This polynomial is the characteristic polynomial of the LFSR in SNOW 2.0 when it is implemented as a $\sigma$-LFSR i.e. when multiplication by $\alpha$ and  $\alpha^{-1}$ are represented by matrices).\par
The keystream has been generated using the following  key (K) and initialization vector (IV).

\begin{equation}
K=[681, 884, 35, 345, 203, 50, 912, 358], \quad
IV = [645, 473, 798, 506]
\end{equation}

\subsubsection{Test Results\\}
\vspace{0.2in}
The results obtained from $14$ NIST tests are shown in table \ref{NISTRand14}. 
 
\begin{table}[H]
\begin{center}
 	\begin{tabular}{|c|c|c|c|}
		\hline
		\textbf{Number}&\textbf{Test}&\textbf{P-Value}&\textbf{Random}\\
		\hline
		01. &Frequency Test (Monobit)                      	& 0.35966689490586123		&\checkmark\\
		\hline
		02.  &Frequency Test within a Block                  &  0.24374184001729746 	 &\checkmark\\
		\hline
		03.  &Run Test                                       & 0.9038184342313019  	 &\checkmark\\
		\hline
		04.  &Longest Run of Ones in a Block                 &	 0.5246846287441829   & \checkmark\\
		\hline
		05.  &Binary Matrix Rank Test                        &	0.1371167998339736     &	 \checkmark\\
		\hline
		06.  &Discrete Fourier Transform (Spectral) Test     &	0.1371167998339736   &	 \checkmark\\
		\hline
		07.  &Non-Overlapping Template Matching Test         & 0.3189818228443801     &  \checkmark\\
		\hline
		08.  &Overlapping Template Matching Test             &	0.211350493609367       &	 \checkmark\\
		\hline
		09.  & Maurer's Universal Statistical test            & 0.4521082097311434	  &\checkmark\\
		\hline
		10.  &Linear Complexity Test                          & 0.1647939201114819 	  &\checkmark\\
		\hline
		11. &Serial Test									  & 0.7821664366290292                     &\checkmark\\
		\hline
		12.  &Approximate Entropy Test                        & 0.880218270580662     &\checkmark\\
		\hline
		13.  &Cummulative Sums (Forward) Test                &	 0.34630799549695923	  &\checkmark\\
		\hline
		14.  &Cummulative Sums (Reverse) Test                &	 0.6633686090204551     &\checkmark\\
		\hline 
		\end{tabular}
\end{center}
	\caption{NIST Randomness Test}
	\label{NISTRand14}
\end{table}

 The results obtained for the \textit{Random Excursions Test} are shown in table \ref{RandExc}.

\begin{table}[H]
\begin{center}
  	\begin{tabular}{|c|c|c|c|}
		\hline
		\textbf{State} &\textbf{CHI SQUARED} &\textbf{P-Value} &\textbf{Random}\\
		\hline	
		-4 & 9.375081555789523 & 0.09500667227464867 & \checkmark\\
		\hline
		-3 & .9066918454935624 & 0.969735280059932 & \checkmark\\
		\hline
		-2 & 3.196312192020347 & 0.6697497097941535 & \checkmark\\
		\hline
		-1 & 5.343347639484978 & 0.3754291967984828 & \checkmark\\
		\hline
		1 & 5.446351931330472 & 0.363864453873992 &  \checkmark\\
		\hline
		2 & 6.937635775976262 & 0.22531988887331122 & \checkmark\\
		\hline
		3 & 13.843145064377687 & 0.016637085511558194 & \checkmark\\
		\hline
		4 & 3.790226890440857 & 0.5799959469559587 & \checkmark\\
		\hline
	\end{tabular}
\end{center}	
	\caption{15. Random Excursions Test}
	\label{RandExc}
	\end{table}

Table \ref{RandExcVar} shows the results for the $16^{th}$ test i.e. the \textit{Random Excursions Variant Test}

\begin{table}[H]
\begin{center}
	\begin{tabular}{|c|c|c|c|c|c|c|c|}
		\hline
		\textbf{State} &\textbf{Count} &\textbf{P-Value} &\textbf{Random}& \textbf{State} &\textbf{Count} &\textbf{P-Value} &\textbf{Random}\\
		\hline
		'-9.0'&  270&  0.11944065987006025& \checkmark &'+1.0'& 467&  0.9738690952237389 & \checkmark\\
		\hline
		'-8.0'&307&0.1787039957218327&\checkmark&'+2.0'&  488 & 0.6773674079894312 &\checkmark\\
		\hline
		'-7.0'& 359 & 0.3310083710716354 & \checkmark&'+3.0'& 470 &0.9532739974827851&\checkmark\\
		\hline
		'-6.0'& 389 &0.44696915370831947 &\checkmark&'+4.0'& 416& 0.5358953455898371& \checkmark\\
		\hline
		'-5.0'& 418 &0.6002107789999439 &\checkmark&'+5.0'& 386& 0.3823929438406025&\checkmark\\
		\hline
		'-4.0'& 426& 0.6204409395957975&\checkmark&'+6.0'& 397& 0.495576078534262 & \checkmark\\
		\hline
		'-3.0'& 439 & 0.6924575808023399&\checkmark&'+7.0'& 430& 0.7436251044167517& \checkmark\\
		\hline
		'-2.0'& 486&  0.7052562223122887&\checkmark&'+8.0'& 454& 0.9191606777606087& \checkmark\\
		\hline
		'-1.0'& 486& 0.512389348919496&\checkmark&'+9.0'& 481& 0.9051424340008056 &\checkmark\\
		\hline
	\end{tabular}
\end{center}
	\caption{Randomness Excursions Variant Test}
	\label{RandExcVar}
\end{table}

These results are comparable to that of SNOW 2.0. Note that the feedback configuration of SNOW 2.0 is one of the possible feedback configurations in the $\sigma$-KDFC scheme.


\subsection{Challenges in implementation}
The main problem of KDFC lies in its software implementation. Since the feedback function is not fixed, look up tables cannot be used in the implementation of the $\sigma$-LFSR. Further, the choice of the feedback configurations is not restricted to the set of efficiently implementable $\sigma$-LFSR configurations given in (\cite{zeng2007high}). This makes the implementation of KDFC SNOW extremely challenging.

 In our implementation, the state of the $\sigma$-LFSR is stored as a set of 32 integers. The $i$-th integer corresponds to the $i$-th output of the delay blocks. Calculating the feedback function of the $\sigma$-LFSR involves calculating the bitwise XOR of a subset of the columns of the feedback matrices ($B_i$s). In order to make the implementation more efficient, for all $1\leq i \leq 32$, the $i$-th columns of the feedback matrices are stored in adjacent memory locations. Thus, each integer in the state of the $\sigma$-LFSR corresponds to a set of columns of the the feedback matrices which are stored in a continuous block of memory. The state vector is now sampled one integer at a time and the columns of the $B_i$s corresponding to the non zero bits in these integers are XORed with each other. We then do a bit-wise right shift on each of these integers and introduce the result of the XOR operation bitwise as the most significant bits. In this way, the $\sigma$-LFSR can be implemented using bitwise XORs and shifts. The FSM is implemented as in SNOW 2.0 \cite{ekdahl2002new}. This implementation takes 25 cycles to generate a single word on an Intel Probook 4440s machine with a 2.8 Ghz i5 processor.
 
 Each iteration of Algorithm 1 involves solving a system of linear equations. This process is time consuming and contributes to increasing the initialization time. The initialization process was implemented using a C code with open mp (with 3 threads). In this implementation linear equations were solved using a parallel implementation of the LU decomposition algorithm. The initialization process took a total time of 5 to 15 seconds on an Intel Probook 4440s machine with a 2.8 Ghz i5 processor.

\section{Conclusions}\label{Conc}
In this paper, we have described a method of using $\sigma$-LFSRs with key dependent feedback configurations in stream ciphers that use word based LFSRs.  In this method, an iterative configuration generation algorithm(CGA) uses key-dependant random numbers to generate a random feedback configuration for the $\sigma$-LFSR. We have theoretically analysed the algebraic degree of the resulting feedback configuration  As a test case, we have demonstrated how this scheme can be used along with the Finite State Machine of SNOW 2.0. We have analysed the security of the resulting  key-stream generator against various attacks and have demonstrated the improvement in security as compared to SNOW 2.0.   Further, the keytreams generated by the proposed method are comparable to SNOW 2.0 from a randomness point of view.

\section*{Acknowledgement}
The authors are grateful to Prof. Harish K. Pillai, Department of Electrical Engineering, Indian Institute of Technology Bombay for his valuable guidance and Associate Prof. Gaurav Trivedi, Department of Electronics and Electrical Engineering, Indian Institute of Technology Guwahati for helping us with computational resources.

\section{Declarations}
\subsection{Funding}
Not applicable.
\subsection{Conflict of Interest}
Authors declare no conflict of interest.
\subsection{Code availability}
Code will be available upon request to the authors.
\subsection{Authors' contributions}
Not applicable.
\subsection{Availability of data and material}
Not applicable.
\subsection{Consent to participate}
Not applicable.
\subsection{Consent for publication }
All authors consent to the publication of the manuscript in Cryptography and Communications upon acceptance.

\bibliographystyle{alpha}
\bibliography{NEWSIGMA}

\newcommand{\etalchar}[1]{$^{#1}$}
\begin{thebibliography}{WBDC03}

\bibitem[AE09]{ahmadi2009heuristic}
Hadi Ahmadi and Taraneh Eghlidos.
\newblock Heuristic guess-and-determine attacks on stream ciphers.
\newblock {\em IET Information Security}, 3(2):66--73, 2009.

\bibitem[AG98]{New-Pr-01}
N.K. Ahmad, A.and~Nanda and K.~Garg.
\newblock Critical role of primitive polynomials in an lfsr based testing
  technique.
\newblock {\em Electronics Letters}, 24(15):953--956, 1998.

\bibitem[ANA16]{Appl-01}
Diyana Afdhila, Surya~Michrandi Nasution, and Fairuz Azmi.
\newblock Implementation of stream cipher salsa20 algorithm to secure voice on
  push to talk application.
\newblock In {\em IEEE Asia Pacific Conference on Wireless and Mobile
  (APWiMob)}, Bandung, Indonesia, September 2016. IEEE.

\bibitem[BBC{\etalchar{+}}08]{berbain2008sosemanuk}
C{\^o}me Berbain, Olivier Billet, Anne Canteaut, Nicolas Courtois, Henri
  Gilbert, Louis Goubin, Aline Gouget, Louis Granboulan, C{\'e}dric Lauradoux,
  Marine Minier, et~al.
\newblock Sosemanuk, a fast software-oriented stream cipher.
\newblock In {\em New stream cipher designs}, pages 98--118. Springer, 2008.

\bibitem[BG05]{billet2005resistance}
Olivier Billet and Henri Gilbert.
\newblock Resistance of snow 2.0 against algebraic attacks.
\newblock In {\em Cryptographers’ Track at the RSA Conference}, pages 19--28,
  San Francisco, Ca, USA, February 2005. Springer.

\bibitem[CHM{\etalchar{+}}04]{SNOW1-04}
Kevin Chen, Matt Henricksen, William Millan, Joanne Fuller, Leonie Simpson,
  Ed~Dawson, Hoon~Jae Lee, and Sang~Jae Moon.
\newblock Dragon: A fast word based stream cipher.
\newblock In {\em International Conference on Information Security and
  Cryptology}, Seoul, Korea (Republic of), December 2004.

\bibitem[CJS00]{chepyzhov2000simple}
Vladimor~V Chepyzhov, Thomas Johansson, and Ben Smeets.
\newblock A simple algorithm for fast correlation attacks on stream ciphers.
\newblock In {\em International Workshop on Fast Software Encryption}, pages
  181--195. Springer, 2000.

\bibitem[DMC09]{New-SN-07}
Blandine Debraize and Irene Marquez~Corbella.
\newblock Fault analysis of the stream cipher snow 3g.
\newblock In {\em Workshop on Fault Diagnosis and Tolerance in Cryptography
  (FDTC)}, pages 105--112, Lausanne, Switzerland, September 2009. IEEE.

\bibitem[EJ00]{SNOW1-01}
Patrick Ekdahl and Thomas Johansson.
\newblock Snow -- a new stream cipher.
\newblock In {\em 1st NESSIE Workshop}, Heverlee, Belgium, November 2000.

\bibitem[EJ02]{ekdahl2002new}
Patrik Ekdahl and Thomas Johansson.
\newblock A new version of the stream cipher snow.
\newblock In {\em International Workshop on Selected Areas in Cryptography},
  pages 47--61, St. John’s, Newfoundland, Canada, August 2002. Springer.

\bibitem[EJMY19]{ekdahl2019new}
Patrik Ekdahl, Thomas Johansson, Alexander Maximov, and Jing Yang.
\newblock A new snow stream cipher called snow-v.
\newblock {\em IACR Transactions on Symmetric Cryptology}, pages 1--42, 2019.

\bibitem[HJB09]{SNOW1-02}
Martin Hell, Thomas Johansson, and Lennart Brynielsson.
\newblock An overview of distinguishing attacks on stream ciphers.
\newblock {\em Cryptography and Communications}, 1(1):71--94, 2009.

\bibitem[JJ01]{jonsson2001correlation}
Fredrik J{\"o}nsson and Thomas Johansson.
\newblock Correlation attacks on stream ciphers over gf (2\^{} n).
\newblock In {\em IEEE International Symposium on Information Theory (ISIT),
  2001}, pages 140--140, 2001.

\bibitem[KP11]{krishnaswamy2011number}
Srinivasan Krishnaswamy and Harish~K Pillai.
\newblock On the number of linear feedback shift registers with a special
  structure.
\newblock {\em IEEE transactions on information theory}, 58(3):1783--1790,
  2011.

\bibitem[KP14]{krishnaswamy2014number}
Srinivasan Krishnaswamy and Harish~K Pillai.
\newblock On the number of special feedback configurations in linear modular
  systems.
\newblock {\em Systems \& Control Letters}, 66:28--34, 2014.

\bibitem[KTS17]{Dyn-Str-01}
Shinsaku Kiyomoto, Toshiaki Tanaka, and Kouichi Sakurai.
\newblock K2: A stream cipher algorithm using dynamic feedback control.
\newblock In {\em International Conference on Security and Cryptography
  (SECRYPT)}, Barcelona, Spain, July 2017.

\bibitem[LLP08]{lee2008cryptanalysis}
Jung-Keun Lee, Dong~Hoon Lee, and Sangwoo Park.
\newblock Cryptanalysis of sosemanuk and snow 2.0 using linear masks.
\newblock In {\em International Conference on the Theory and Application of
  Cryptology and Information Security}, pages 524--538. Springer, 2008.

\bibitem[Mat93]{matsui1993linear}
Mitsuru Matsui.
\newblock Linear cryptanalysis method for des cipher.
\newblock In {\em Workshop on the Theory and Application of of Cryptographic
  Techniques}, pages 386--397. Springer, 1993.

\bibitem[MS89]{meier1989fast}
Willi Meier and Othmar Staffelbach.
\newblock Fast correlation attacks on certain stream ciphers.
\newblock {\em Journal of cryptology}, 1(3):159--176, 1989.

\bibitem[NP14]{nia2014new}
Mohammad Sadegh~Nemati Nia and Ali Payandeh.
\newblock The new heuristic guess and determine attack on snow 2.0 stream
  cipher.
\newblock {\em IACR Cryptology ePrint Archive}, 2014:619, 2014.

\bibitem[NW06]{nyberg2006improved}
Kaisa Nyberg and Johan Wall{\'e}n.
\newblock Improved linear distinguishers for snow 2.0.
\newblock In {\em International Workshop on Fast Software Encryption}, pages
  144--162, Graz, Austria, March 2006. Springer.

\bibitem[PBC19]{Appl-02}
Maxime Pistono, Reda Bellafqira, and Gouenou Coatrieux.
\newblock Secure processing of stream cipher encrypted data issued from iot:
  Application to a connected knee prosthesis.
\newblock In {\em 41st Annual International Conference of the IEEE Engineering
  in Medicine and Biology Society (EMBC)}, Berlin, Germany, July 2019. IEEE.

\bibitem[ros]{rose1999t}


\bibitem[WBDC03]{watanabe2003distinguishing}
Dai Watanabe, Alex Biryukov, and Christophe De~Canniere.
\newblock A distinguishing attack of snow 2.0 with linear masking method.
\newblock In {\em International Workshop on Selected Areas in Cryptography},
  pages 222--233. Springer, 2003.

\bibitem[Xt11]{xiu2011zuc}
FENG Xiu-tao.
\newblock Zuc algorithm: 3gpp lte international encryption standard.
\newblock {\em Information Security and Communications Privacy}, 12, 2011.

\bibitem[ZHH07]{zeng2007high}
Guang Zeng, Wenbao Han, and Kaicheng He.
\newblock High efficiency feedback shift register: sigma-lfsr.
\newblock {\em IACR Cryptology ePrint Archive}, 2007:114, 2007.

\bibitem[ZXM]{zhang2015fast}
Bin Zhang, Chao Xu, and Willi Meier.
\newblock Fast correlation attacks over extension fields, large-unit linear
  approximation and cryptanalysis of snow 2.0.

\end{thebibliography}
\end{document}